\documentclass[conference]{IEEEtran}
\IEEEoverridecommandlockouts
\usepackage{multirow}
\usepackage{subfigure}
\usepackage[linesnumbered,ruled, vlined]{algorithm2e}
\usepackage{amsmath, amsthm,amsfonts, bm}
\usepackage{array}
\usepackage{balance}
\usepackage{enumitem}
\usepackage{soul}
\usepackage{tablefootnote}
\usepackage{threeparttable}
\usepackage{footnote}
\usepackage{xcolor}
\usepackage{romannum}
\makesavenoteenv{tabular}
\makesavenoteenv{table}
\usepackage{url}
\usepackage{graphicx}
\usepackage{marginnote}

\newcolumntype{L}[1]{>{\raggedright\let\newline\\\arraybackslash\hspace{0pt}}m{#1}}
\newcolumntype{C}[1]{>{\centering\let\newline\\\arraybackslash\hspace{0pt}}m{#1}}
\newcolumntype{R}[1]{>{\raggedleft\let\newline\\\arraybackslash\hspace{0pt}}m{#1}}

\newtheorem{problem}{Problem}

\newtheorem{definition}{Definition}
\newtheorem{theorem}{Theorem}
\newtheorem{lemma}{Lemma}

\AtBeginDocument{%
  \providecommand\BibTeX{{%
    \normalfont B\kern-0.5em{\scshape i\kern-0.25em b}\kern-0.8em\TeX}}}

\newcommand{\kaixin}{\color{black}}
\newcommand{\cheng}{\color{black}}

\newcommand{\kaixinkdd}{\color{black}}
\newcommand{\yui}{\color{black}}
\newcommand{\chengB}{\color{black}}
\newcommand{\kaixinC}{\color{black}}
\newcommand{\chengc}{\color{black}}

\begin{document}
\sloppy


\author{
\IEEEauthorblockN{
Kaixin Wang\IEEEauthorrefmark{1}, 
Kaiqiang Yu\IEEEauthorrefmark{2}, 
Cheng Long\IEEEauthorrefmark{2}
}
\IEEEauthorblockA{
\IEEEauthorrefmark{1}College of Computer Science, Beijing University of Technology, Beijing, China\\
\IEEEauthorrefmark{2}College of Computing and Data Science, Nanyang Technological University, Singapore\\
kaixin.wang@bjut.edu.cn,
kaiqiang002@e.ntu.edu.sg,
c.long@ntu.edu.sg
 }
}

\title{Maximal Clique Enumeration with Hybrid Branching and Early Termination}

\maketitle

\begin{abstract}
    Maximal clique enumeration (MCE) is crucial for tasks like community detection and biological network analysis. Existing algorithms typically adopt the branch-and-bound framework with the vertex-oriented Bron-Kerbosch (BK) branching strategy, which forms the sub-branches by expanding the partial clique with a vertex. 
    {\chengc In this paper, we present} a novel approach, \texttt{HBBMC}, a hybrid framework combining vertex-oriented BK branching and edge-oriented BK branching, where the latter adopts a branch-and-bound framework which forms the sub-branches by expanding the partial clique with a edge. 
    This hybrid strategy enables more effective pruning {\chengc and helps achieve a worst-case time complexity better than the best-known one under a condition which holds for the majority of real-world graphs.}
    To further enhance efficiency, we introduce an early termination technique, which leverages the topological information of the graphs and constructs the maximal cliques directly without branching. 
    {\chengc Our early termination technique is} applicable to all branch-and-bound frameworks. 
    Extensive experiments 
    demonstrate the superior performance of our techniques. 
\end{abstract}

\section{Introduction}
\label{sec:intro}

Given a graph $G$, a clique is a complete subgraph of $G$ where any two vertices are connected by an edge. A clique is said to be \emph{maximal} if it is not a subgraph of another clique. Maximal clique enumeration (MCE), which is to enumerate all maximal cliques in a graph, is a fundamental graph problem and finds diverse applications across different networks, e.g., social networks, biological networks and Web networks. Some notable applications include the detection of the communities in social networks \cite{lu2018community, wen2016maximal}, the prediction of the functions of {\cheng proteins} in biological networks (where the vertices are the {\cheng proteins} and the edges are the interactions between proteins) \cite{yu2006predicting, matsunaga2009clique}. Moreover, maximal clique enumeration has 
{\cheng been} applied to other data mining tasks including document clustering \cite{augustson1970analysis}, image analysis and recovery \cite{horaud1989stereo}, and the discovery of frequent purchased item patterns in e-commerce \cite{zaki1997new}.

Quite a few algorithms have been developed to solve {\cheng the} MCE problem \cite{tomita2006worst, eppstein2010listing, eppstein2013listing, li2019fast, naude2016refined, chang2013fast, conte2016sublinear, deng2023accelerating}, where most of these approaches adopt the \emph{branch-and-bound} (BB) framework \cite{tomita2006worst, eppstein2010listing, eppstein2013listing, li2019fast, naude2016refined, deng2023accelerating}. The framework involves a recursive procedure, which recursively partitions the search space (i.e., the set of all maximal cliques) into multiple sub-spaces via a conventional and widely-used branching strategy called Bron-Kerbosch (BK) branching \cite{bron1973algorithm}. Specifically, each branch $B$ is represented by a {\kaixin triplet} $(S, g_C, g_X)$, where the vertex set $S$ induces a partial clique that has been found so far, $g_C$ is the candidate graph where each vertex inside can be added to $S$ to form a larger clique and $g_X$ is the exclusion
graph where each vertex inside must be excluded from any maximal clique
within the branch $B$. To enumerate the maximal cliques within $B$, a
set of sub-branches is created through branching operations. Each sub-branch expands $S$ 
{\chengc with}
a vertex $v$ from $g_C$, and updates $g_C$ and $g_X$ accordingly (by removing {\yui from $g_C$ and $g_X$ those} vertices that are not connected with $v$, respectively). This recursive process continues until $g_C$ has no vertex inside, at which point $S$ is reported as a maximal clique if $g_X$ also has no vertex inside; otherwise $S$ is not a maximal clique. Once the recursive process 
{\cheng finishes}
the branching step at $v$, the vertex $v$ is excluded from $S$ and moves to $g_X$. The original MCE problem on $G$ can be solved {\chengB by starting} from the initial branch $B=(S, g_C, g_X)$ with {\kaixinC $S$ being an empty set, $g_C$ being the graph $G$ and $g_X$ being an empty graph}. Since each branching step is conducted upon a vertex, we call it the \emph{vertex-oriented BK branching}. Correspondingly, we 
{\chengc call}
the \underline{v}ertex-oriented BK branching based \underline{BB} framework for \underline{m}aximal \underline{c}lique enumeration as \texttt{VBBMC}.

Existing studies target two directions to enhance the performance. One {\cheng direction} is to prune more sub-branches effectively \cite{tomita2006worst, naude2016refined}. A common practice is the \emph{pivoting} strategy. Given a branch $B=(S,g_C,g_X)$, instead of producing sub-branches {\cheng on} all vertices in $g_C$, the strategy first chooses a pivot vertex $u\in V(g_C)\cup V(g_X)$. Then the branches produced 
{\chengc at}
$u$'s neighbors in $g_C$ can be safely pruned. The rationale is that those maximal cliques that have $u$’s neighbors inside either include $u$ ({\chengB in this case, they} can be enumerated from the branch produced at $u$) or include at least one $u$’s non-neighbors ({\chengB in this case, they} can be enumerated from the
branch produced {\cheng at} $u$’s non-neighbors). 
The other {\cheng direction} is to generate candidate graph instances with smaller {\cheng sizes} \cite{eppstein2010listing, eppstein2013listing, xu2014distributed}. When a vertex $v_i$ is added to $S$, the size of the candidate graph instance shrinks as the candidate graph instance in the sub-branch {\yui which} only considers $v_i$'s neighbors. Then it becomes possible to generate a set of sub-branches with 
{\chengc as shrunk}
candidate graph instances {\chengB as possible} by exploring different vertex orderings. For example, when the degeneracy ordering of $G$ is adopted at the universal branch, the maximum size of a sub-branch can be bounded by the degeneracy number $\delta$ of $G$, which is much smaller than the number of vertices $n$ in $G$ (i.e., $\delta\ll n$) \cite{eppstein2010listing, eppstein2013listing}. 
By 
{\chengc employing}
the above two techniques, existing studies \cite{eppstein2010listing, eppstein2013listing} achieve the state-of-the-art time complexity of $O(n\delta\cdot 3^{\delta/3})$. 
{\chengc We note that while}
some recent research has improved the performance practically \cite{deng2023accelerating, li2019fast, jin2022fast}, {\chengB they do not achieve} a better theorectical {\cheng result}.

Recently, we observe that another branching strategy called edge-oriented branching \cite{wang2024technical} has been 
{\cheng proposed for}
a related problem, namely $k$-clique
listing \cite{chiba1985arboricity}, which aims to list all $k$-cliques in a graph $G$ {\chengc (a $k$-clique is a clique with $k$ vertices)}. Similarly, {\chengc in this branching,} each branch $B$ 
is 
associated with a partial clique $S$ and a candidate graph $g$, whose definitions are the same as those {\chengc for the MCE problem}. At a branching step,
edge-oriented branching constructs a sub-branch by simultaneously including two vertices that are connected with each other (i.e., an edge) from $g$ to $S$. And the candidate graph in the sub-branch is updated to the subgraph of $g$ induced by the common neighbors of the vertices incident to the newly added edge. This strategy takes additional information of two vertices of an edge into account, which facilitates the formation of sub-branches with smaller {\cheng candidate} graphs. A nice property is that, by adopting a carefully specified ordering of edges (i.e., truss-based edge ordering \cite{wang2024technical}) at the initial branch, the largest size of a graph instance among the sub-branches can be bounded {\yui by} $\tau$, whose value is strictly smaller than the degeneracy $\delta$ of the graph $G$ (i.e., $\tau<\delta$).

{\chengB In this paper, we first}
migrate the edge-oriented branching strategy to {\cheng the} MCE problem. 
Given a branch $B=(S,g_C, g_X)$, 
{\chengc each}
sub-branch 
expands $S$ 
{\chengc with}
two vertices of an edge $e$ from $g_C$ {\chengB and} updates both $g_C$ and $g_X$ to be their subgraphs induced by the common neighbors of the two vertices incident to the edge $e$, respectively. Once the recursive process returns from the branching step at $e$, the edge $e$ is excluded from $S$ and moves to $g_X$. We call the above procedure \emph{edge-oriented BK branching}, and 
{\chengc call}
the corresponding framework for MCE problem {\chengc as} \texttt{EBBMC}. While \texttt{EBBMC} inherits the advantage that the {\cheng sizes} of the candidate graph instances have a tigher upper bound $\tau$ with {\cheng the} truss-based edge ordering, it would explore $O(m\cdot 2^{\tau})$ sub-branches in the worst case, which undermines its superiority. 
To tackle this issue, we propose a hybrid branching framework, which 
jointly 
{\chengc applies}
the vertex-oriented branching and the edge-oriented branching 
so that more branches can be pruned. We show a BB algorithm that is based on the hybrid branching framework, which we call \texttt{HBBMC}, would have a worst-case time complexity of $O(\delta m + \tau m \cdot 3^{\tau/3})$. Compared with the existing algorithms, when the condition $\delta \ge \tau + \frac{3}{\ln 3}\ln \rho$ is satisfied, where $\rho=\frac{m}{n}$ is the edge density of the graph, our algorithm is asymptotically better than the state-of-the-art algorithms. We note that the condition holds for many
real-world graphs, {\chengB as we verify empirically}.

To further enhance the efficiency of BB frameworks, we develop an early termination technique. The key idea is based on an observation that it would be efficient
to construct the maximal cliques directly {\chengB for} dense subgraphs without branching. For example, given a branch $B=(S,g_C,g_X)$ where the candidate graph $g_C$ is a clique and {\kaixin $g_X$ is an empty graph}, it is unnecessary to create sub-branches; instead, we can include all vertices in $g_C$ into $S$ directly.  
{\cheng Similarly,}
when $g_C$ is a 2-plex or 3-plex {\chengc (a graph $g$ is said to be a $t$-plex if each vertex $v\in V(g)$ has at most $t$ non-neighbors including itself)}, we can also efficiently construct the maximal cliques of $g_C$ in nearly optimal time. 
%
{\chengc Specifically,}
we construct {\chengB a} complementary graph $\overline{g}_C$ of $g_C$, where the complementary graph $\overline{g}_C$ has the same set of vertices as $g_C$, with an edge between two vertices in $\overline{g}_C$ if and only if they are disconnected from each other in $g_C$. An interesting observation is that the inverse graph of a 2-plex or a 3-plex only has isolated vertices, simple paths or simple cycles. By using the topological information of the inverse graph $\overline{g}_C$, we can output all maximal cliques within the branch in the time proportional to the number of maximal cliques inside.
We note that the early termination technique is applicable to all BB frameworks (i.e., \texttt{VBBMC}, \texttt{EBBMC} and \texttt{HBBMC}) without impacting the worst-case time complexity of the BB frameworks.


\smallskip\noindent
\textbf{Contributions.} We summarize our contributions as follows.

\begin{itemize}[leftmargin=*]
    \item We propose a new algorithm \texttt{HBBMC} with a hybrid framework for the MCE problem, which integrates the vertex-oriented and the edge-oriented BK branching strategies. \texttt{HBBMC} has the worst-case time complexity of $O(\delta m + \tau m \cdot 3^{\tau/3})$. When $\delta \ge \tau + \frac{3}{\ln 3}\ln \rho$, which holds for many real-world graphs, \texttt{HBBMC} is asymptotically better than the state-of-the-art algorithms. (Section~\ref{sec:BBMC}) 
    
    \item {We further develop an early termination technique to enhance the efficiency of branch-and-bound frameworks including \texttt{HBBMC}, i.e., for branches with dense candidate graphs (e.g., $t$-plexes) and empty exclusion graphs, we develop algorithms to directly construct the maximal cliques based on the topological information of the inverse graphs, which would be faster. (Section~\ref{sec:early-termination})}
    
    \item We conduct extensive experiments on both real and synthetic graphs, and the results show that our \texttt{HBBMC} based algorithm with the early termination technique consistently outperforms the state-of-the-art (\texttt{VBBMC} based) algorithms. (Section~\ref{sec:exp}) 
\end{itemize}

{\chengc In addition,} Section~\ref{sec:problem} reviews the problem and presents some preliminaries. 
Section~\ref{sec:related} reviews the related work and Section~\ref{sec:conclusion} concludes the paper.

\section{Problem and Preliminaries}
\label{sec:problem}

We consider an \emph{unweighted} and \emph{undirected} simple graph $G=(V,E)$, where $V$ is the set of vertices and $E$ is the set of edges. We denote by $n=|V|$ and $m=|E|$ the cardinalities of $V$ and $E$, respectively. Given $u,v\in V$, both $(u,v)$ and $(v,u)$ denote the undirected edge between $u$ and $v$. Given $u\in V$, we denote by $N(u,G)$ the set of neighbors of $u$ in $G$, i.e., $N(u,G)=\{v\in V\mid (u,v)\in E\}$ and define $d(u, G)=|N(u,G)|$.
Given $V_{sub}\subseteq V$, we use $N(V_{sub},G)$ to denote the common neighbors of vertices in $V_{sub}$, i.e., $N(V_{sub},G)=\{v\in V\mid \forall u\in V_{sub},\ (v,u)\in E\}$.
%
Given $V_{sub}\subseteq V$, we denote by $G[V_{sub}]$ the subgraph of $G$ induced by $V_{sub}$, i.e., $G[V_{sub}]$ includes the set of vertices $V_{sub}$ and the set of edges $\{(u,v)\in E\mid u, v\in V_{sub}\}$.
Given $E_{sub}\subseteq E$, we denote by $G[E_{sub}]$ the subgraph of $G$ induced by $E_{sub}$, i.e., $G[E_{sub}]$ includes the set of edges $E_{sub}$ and the set of vertices $\{ v\in V \mid (v, \cdot) \in E_{sub} \}$.
Let $g$ be a subgraph of $G$ induced by either a vertex subset of $V$ or an edge subset of $E$. We denote by $V(g)$ and $E(g)$ its set of vertices and its set of edges, respectively.
%

\begin{definition}[Clique~\cite{erdos1935combinatorial}]
A subgraph $g$ is said to be a clique if there exists an edge between every pair of vertices, i.e., $E(g)=\{(u,v)\mid u,v\in V(g), {\cheng u\neq v} \}$.
\end{definition}

%

Given a graph $G$, a clique $C$ is said to be maximal if there is no other clique $C'$ containing $C$, i.e., $V(C)\subseteq V(C')$. 

\begin{problem}
[Maximal Clique Enumeration~\cite{bron1973algorithm}]
Given a graph $G=(V,E)$ , the Maximal Clique Enumeration (MCE) problem aims to find all maximal cliques in $G$.
\end{problem}

{\kaixinC The maximal clique enumeration (MCE) problem is hard to solve since the number of maximal cliques in a graph could be exponential (e.g., $3^{n/3}$) in the worst case \cite{moon1965cliques}.}

\section{Branch-and-bound Algorithms}
\label{sec:BBMC}

\subsection{The Existing Branch-and-bound {\cheng Framework}: \texttt{VBBMC}}
\label{subsec:VBBMC}

We first review a suite of \emph{branch-and-bound} algorithms for maximal clique enumeration \cite{tomita2006worst, eppstein2010listing, jin2022fast, li2019fast, deng2023accelerating, naude2016refined, xu2014distributed}, which adopt a 
branching strategy called Bron-Kerbosch (BK) branching \cite{bron1973algorithm}. 
It recursively partitions the search space (i.e., the set of all maximal cliques in $G$) into several sub-spaces via \emph{branching} until some stopping {\cheng criteria} are satisfied \cite{li2019fast, jin2022fast, deng2023accelerating}. The rationale is that a maximal clique can be constructed by merging two smaller cliques, namely a partial clique $S$ and another clique {\cheng which} satisfies the maximality condition. Specifically, each branch $B$ can be represented as a triplet $(S, g_C, g_X)$, where 

\begin{itemize}[leftmargin=*]
    \item \textbf{set $S$} induces a partial clique found so far,
    \item \textbf{subgraph $g_C$} is induced by all candidate vertices that can form a larger clique with the vertices in $S$, and 
    \item \textbf{subgraph $g_X$} is induced by all {\kaixin exclusion}
    vertices that 
    {\cheng are excluded for forming cliques with $S$.}
\end{itemize}

\begin{algorithm}[t]
\caption{The vertex-oriented branching-based BB framework for maximal clique enumeration: \texttt{VBBMC}}
\label{alg:VBBMC}
\KwIn{A graph $G=(V,E)$.}
\KwOut{All maximal cliques within $G$.}
\SetKwFunction{List}{\textsf{VBBMC\_Rec}}
\SetKwProg{Fn}{Procedure}{}{}
\List{$\emptyset, G, G[\emptyset]$}\;
\Fn{\List{$S, g_C, g_X$}} {
    \If{$V(g_C) \cup V(g_X)=\emptyset$}{
        \textbf{Output} a maximal clique $S$;
        \Return\;
    }
    Let $g=G[V(g_C)\cup V(g_X)]$ and choose a pivot $u$ from $V(g)$\;
    \For{each $v_i\in V(g_C)$ based on a given vertex ordering}{
        \lIf{$(u,v_i)\in E(g)$}{\textbf{continue}}
        Create branch $B_i=(S_i,g_{C_i},g_{X_i})$ based on Eq.~(\ref{eq:vbb-branching})\;
        $g_{C_i}',~g_{X_i}'\gets$ Refine $g_{C_i}$ and $g_{X_i}$ by removing those vertices that are not adjacent to $S_i$\;
        \List{$S_i,g_{C_i}',g_{X_i}'$}\;
    }
}
\end{algorithm}

Given a graph $G=(V,E)$, the algorithm starts from the universal search space with $S=\emptyset$, $g_C = G$ and $g_X = G[\emptyset]$. Consider the branching step at a branch $B=(S, g_C, g_X)$. Let $\langle v_1, v_2, \cdots, v_{|V(g_C)|} \rangle$ be an arbitrary ordering of the vertices in $g_C$. The branching step would produce $|V(g_C)|$ sub-branches, where the $i$-th branch, denoted by $B_i=(S_i, g_{C_i}, g_{X_i})$, includes $v_i$ to $S$ and excludes $\{v_1, v_2, \cdots, v_{i-1}\}$. 
Formally, for $1\le i \le |V(g_C)|$, we have 
\begin{equation}
\label{eq:vbb-branching}
\begin{aligned}
    S_i&=S\cup \{v_i\}, \quad
    g_{C_i} = G[V(g_C)\setminus\{v_1, v_2 \cdots ,v_{i}\}], \\
    g_{X_i} &= G[V(g_X)\cup\{v_1, v_2\cdots, v_{i-1}\}]
\end{aligned}
\end{equation}
We note that $g_{C_i}$ and $g_{X_i}$ can be further refined by removing those vertices that are not adjacent to $S_i$ since they cannot form any cliques with $S_i$ within the search space. 
The process continues until no vertex can be further included into $S$, i.e., $V(g_C)=\emptyset$, at which point $S$ is a maximal clique if $V(g_X)$ is also empty; otherwise, $S$ is not maximal. 
We call the above branching step \emph{vertex-oriented BK branching} since each sub-branch is formed by including a vertex to the set $S$. 
We 
{\chengc call}
the {\chengB \underline{v}ertex-oriented \underline{b}ranch-and-\underline{b}ound framework for \underline{MC}E} 
{\chengc as} \texttt{VBBMC}. 
Furthermore, during the recursive process, the state-of-the-art algorithms \cite{tomita2006worst,eppstein2010listing,eppstein2013listing, xu2014distributed} usually adopt a pruning technique, called \emph{pivoting} strategy, {\cheng for pruning some branches}. Formally, given a branch $B=(S,g_C,g_X)$, instead of producing branches for all vertices in $g_C$, these algorithms first choose a vertex $v\in V(g_C)\cup V(g_X)$, which is called the \emph{pivot vertex}. They then create the branches only for $v$'s non-neighbors in $g_C$ (including $v$ itself if $v\in V(g_C)$), i.e., the branches produced on $v$'s neighbors are pruned. The rationale is that those maximal cliques that have $v$'s neighbors inside either include $v$ 
or include at least one {\cheng of} $v$'s non-neighbors, 
since otherwise, the clique is not maximal. 

We present the pseudo-code in Algorithm~\ref{alg:VBBMC}. We remark that Algorithm~\ref{alg:VBBMC} only presents the algorithmic idea while the possible pruning techniques (e.g., detailed pivoting selection \cite{tomita2006worst, naude2016refined, jin2022fast} and graph reduction \cite{deng2023accelerating}) and detailed implementations (e.g., the data structures used for representing a branch \cite{eppstein2010listing, eppstein2013listing, li2019fast}) vary in existing algorithms. 
{\kaixinkdd {\chengB Different variants of \texttt{VBBMC} have different time complexities (which we summarize in the Appendix~\ref{sec:app-vbb-time}), among which the best time complexity is
$O(n\delta \cdot 3^{\delta/3})$}, where $\delta$ is the degeneracy of the graph \cite{eppstein2010listing, eppstein2013listing}. }
Some more details of variants of \texttt{VBBMC} will be provided in the related work (Section~\ref{sec:related}).

\subsection{A New Branch-and-bound Framework: \texttt{EBBMC}}
\label{subsec:EBBMC}

{\chengB We} observe that there exists another {\cheng branching strategy} in which, given a branch $B=(S, g_C, g_X)$, the sub-branches can be created based on the \emph{edges} in $g_C$~\cite{wang2024technical}. We 
{\chengc call}
the branching step involved in this approach 
{\chengc as}
\emph{edge-oriented branching}. Specifically, consider the branching step at $B$ and let $\langle e_1, e_2, \cdots, e_{|E(g_C)|}\rangle$ be an arbitrary edge ordering {\cheng for the edges in} $g_C$. Then the branching step would create $|E(g_C)|$ sub-branches from $B$, where the $i$-th branch, denoted by $B_i=(S_i, g_{C_i}, g_{X_i})$, includes the two vertices incident to the edge $e_i$, i.e., $V(e_i)$, and excludes from $g_C$ those edges in $\{e_1, e_2, \cdots, e_{i-1}\}$.
Formally, for $1\le i\le |E(g_C)|$, we have
\begin{equation}
\label{eq:ebb-branching}
\begin{aligned}
    S_i &= S \cup V(e_i),\quad g_{C_i} = G[E(g_C)\setminus\{e_1, e_2, \cdots, e_{i}\}],\\
    g_{X_i} &= G[E(g_X)\cup \{e_1, e_2, \cdots, e_{i-1}\}],
\end{aligned}
\end{equation}
We note that $g_{C_i}$ and $g_{X_i}$ can be further refined by removing those vertices that are not adjacent to $S_i$.
%
{\chengB We note that}
given a branch $B$, the branching steps formulated by Eq.~(\ref{eq:ebb-branching}) do not always cover all maximal cliques in $B$ since there may exist some vertices whose degrees are zero in $g_C$ {\chengc (i.e., the vertices can be added to $S$ to form larger cliques, yet this cannot be achieved by expanding $S$ with any edges in $g_C$)}. 
%
{\kaixin
For example, when there are isolated vertices in $G$, they are also maximal cliques, i.e., 1-cliques. However, any branch produced by Eq.~(\ref{eq:ebb-branching}) does not cover these maximal cliques. In general, for the maximal cliques with their {\chengB sizes} being odd numbers, each edge-oriented branching step {\cheng would} involve two vertices to $S$ and finally there exists a branch $B=(S,g_C,g_X)$ in which there are isolated vertices in $g_C$.}
To enumerate all maximal cliques which include these vertices, we observe that no other vertices can be further added to $S$ after any of these zero-degree vertices is included in $S$. Thus, we can directly output these maximal cliques if they satisfy the maximality condition. 
Formally, let $Z(g_C)$ be the set of vertices in $g_C$ whose degrees are zero and let $\langle v_1, v_2, \cdots, v_{|Z(g_C)|} \rangle$ be an arbitrary vertex ordering. We create additional $|Z(g_C)|$ branches. For each entry $1\le j\le |Z(g_C)|$, the branch $B_j=(S_j, g_{C_j}, g_{X_j})$ is represented by
\begin{equation}
\label{eq:zbb-branching}
\begin{aligned}
    S_j &= S\cup \{v_j\}, \quad g_{C_j} = G[\emptyset], \\
    g_{X_j} &= \big(V_{C\cup X} \cup \{v_1, v_2, \cdots, v_{j-1}\}, E_{C\cup X}\big)
\end{aligned}
\end{equation}
We note that (1) $V_{C\cup X}$ (resp. $E_{C\cup X}$) is the vertex (resp. edge) set of the graph $g$, where $g$ is the subgraph of $G$ induced by the edge set $E(g_C)\cup E(g_X)$, (2) we can choose an arbitrary vertex ordering for the vertices in $Z(g_C)$ since the subgraph $g_{C_j}$ has no vertices or edges inside and (3) $g_{X_j}$ can be further refined by removing those vertices that are not adjacent to $S_j$. Then, if the graph $g_{X_j}$ has no vertex inside after refinement, $S_j$ induces a maximal clique; otherwise, $S_j$ is not a maximal clique. 
In summary, with Eq.~(\ref{eq:ebb-branching}) and Eq.~(\ref{eq:zbb-branching}), {\chengc $|E(g_C)| + |Z(g_C)|$ branches would be formed, and} all maximal cliques in $B$ would be enumerated exactly once.


\begin{algorithm}[t]
\caption{The edge-oriented branching-based BB framework for maximal clique enumeration: \texttt{EBBMC}}
\label{alg:EBBMC}
\KwIn{A graph $G=(V,E)$.}
\KwOut{All maximal cliques within $G$.}
\SetKwFunction{List}{\textsf{EBBMC\_Rec}}
\SetKwProg{Fn}{Procedure}{}{}
\List{$\emptyset, G, G[\emptyset]$}\;
\Fn{\List{$S, g_C, g_X$}} {
    \If{$V(g_C) \cup V(g_X)=\emptyset$}{
        \textbf{Output} a maximal clique $S$;
        \Return\;
    }

    \For{each $e_i\in E(g_C)$ based on a given edge ordering}{
        Create branch $B_i=(S_i,g_{C_i},g_{X_i})$ based on Eq.~(\ref{eq:ebb-branching})\;
        $g_{C_i}',~g_{X_i}'\gets$ Refine $g_{C_i}$ and $g_{X_i}$ by removing those vertices that are not adjacent to $S_i$\;
        \List{$S_i,g_{C_i}',g_{X_i}'$}\;
    }

    \For{each $v_j \in Z(g_C)$}{
        Create branch $B_j=(S_j, g_{C_j}, g_{X_j})$ based on Eq.~(\ref{eq:zbb-branching})\;
        $g_{X_j}'\gets$ Refine $g_{X_j}$ by removing those vertices that are not adjacent to $S_j$\;
        \If{$V(g_{X_j}') = \emptyset$}{
            \textbf{Output} a maximal clique $S_j$
        }
    }
}
\end{algorithm}

We call the above branching strategy {\chengc as} \emph{edge-oriented BK branching} {\chengc since} it {\cheng is} mainly based on the edge-oriented branching strategy to {\cheng partition} a search space. 
We 
{\chengc call}
the {\chengB \underline{e}dge-oriented \underline{b}ranch-and-\underline{b}ound framework for \underline{MC}E}
{\chengc as}
\texttt{EBBMC}. The pseudo-code is shown in Algorithm~\ref{alg:EBBMC}. 
{\kaixin
Compared with edge-oriented branching strategy for $k$-clique listing problem (i.e., \texttt{EBBkC} framework) \cite{wang2024technical}, both frameworks are associated with a partial clique $S$ and a candidate subgraph, whose definitions are the same with each other.
{\kaixinkdd 
The major difference 
is that in \texttt{EBBMC}, the maximal cliques can have a range of sizes, varying from the smallest one (e.g., a vertex, representing a maximal 1-clique) to the largest one ($\omega$-clique, where $\omega$ denotes the size of the maximum clique) while in \texttt{EBBkC}, all listed cliques are required to possess exactly $k$ vertices inside, with the value of $k$ being a pre-determined input.}
{\chengB As a result, for \texttt{EBBkC}, the branching step would still be conducted based on the edges as long as $|S|+1<k$, no matter whether there exist isolated vertices in the candidate subgraph of the branch. 
{\chengc Therefore,}
the branching procedure represented in Eq.~(\ref{eq:zbb-branching}) is not necessary for \texttt{EBBkC}.}



}

\noindent
\textbf{Ordering of the edges.} 
%
For a sub-branch $B_i$ produced from $B$, the size of the graph $g_{C_i}$ after the refinement (i.e., the number of the vertices in $g_{C_i}'$) is equal to the number of common neighbors of vertices in $V(e_i)$ in ${g}_{C_i}$ (line 7 of the Algorithm~\ref{alg:EBBMC}). That is
\begin{equation}
|V(g_{C_i}')| = |N(V(e_i), {g}_{C_i})|    
\end{equation}
where ${g}_{C_i}$ is defined in Eq.~(\ref{eq:ebb-branching}). 
To minimize the size of the graph instances for the sub-branches, as it would reduce the time costs for solving the current branch, we choose the \emph{truss-based edge ordering} \cite{wang2024technical} for branching. Specifically, the truss-based edge ordering can be obtained via a greedy procedure: it iteratively removes from the remaining $g_C$ the edge whose vertices have the smallest number of common neighbors and then adds it to the end of the ordering. By adopting the truss-based edge ordering, the largest size of a graph instance among all sub-branches created during the recursive procedure can be bounded by $\tau$, whose value is strictly smaller than the degeneracy number $\delta$ of the graph $G$ (i.e., $\tau<\delta$) \cite{wang2024technical}. 
The pseudo-code of \texttt{EBBMC} with truss-based edge ordering is presented in Algorithm~\ref{alg:EBBMC-T}. In particular, we only compute the truss-based edge ordering of $G$ (i.e., $\pi_{\tau}(G)$) for the branching at the initial branch (lines 2-5). Then, for any other branching step at a following branch $B=(S,g_C,g_X)$, edges in $g_C$ adopt the same ordering to those used in $\pi_{\tau}(G)$, which could differ with the truss-based edge ordering of $g_C$. 
Given a branch $B$ other than the initial branch, to create the sub-branches efficiently, we maintain four additional auxiliary sets for each edge in $G$, i.e., $V_+(\cdot)$, $E_+(\cdot)$, $V_-(\cdot)$ and $E_-(\cdot)$ (line 4). The idea is that for an edge $e$, all edges in $E_+(e)$ (resp. $E_-(e)$) are ordered behind (resp. before) $e$ in $\pi_{\tau}$ and the vertices incident to these edges are both connected with those incident to $e$, and $V_+(e)$ (resp. $V_-(e)$) is the set of the vertices in $G[E_+(e)]$ (resp. $G[E_-(e)]$). Therefore, the branching steps can be efficiently conducted via
set intersections (line 9). The correctness of this implementation can be easily verified.

\begin{algorithm}[t]
\caption{\texttt{EBBMC} with truss-based edge ordering}
\label{alg:EBBMC-T}
\KwIn{A graph $G=(V,E)$.}
\KwOut{All maximal cliques within $G$.}
\SetKwFunction{List}{\textsf{EBBMC-T\_Rec}}
\SetKwProg{Fn}{Procedure}{}{}
$\pi_{\tau}(G)\gets$ the truss-based ordering of edges in $G$\;
\For{each edge $e_i\in E(G)$ following $\pi_{\tau}(G)$}{
    Obtain $S_i$, $g_{C_i}'$ and $g_{X_i}'$ (lines 6-7 of Algorithm~\ref{alg:EBBMC})\;
    Initialize $V_+(e_i)\gets V(g_{C_i}')$, $E_+(e_i)\gets E(g_{C_i}')$, $V_-(e_i)\gets V(g_{X_i}')$ and $E_-(e_i)\gets E(g_{X_i}')$\;
    \List{$S_i,g_{C_i}',g_{X_i}'$}\;
}
Conduct Termination (lines 9-12 of Algorithm~\ref{alg:EBBMC})\;
\Fn{\List{$S, g_C, g_X$}} {
    Conduct Termination (lines 3-4 and 9-12 of Algorithm~\ref{alg:EBBMC});
    
    \For{each edge $e\in E(g_C)$}{
        $S'\gets S\cup V(e)$,~$g_C'\gets \big(V(g_C)\cap V_+(e), E(g_C)\cap E_+(e)\big)$ and $g_X'\gets \big(V(g_X)\cap V_-(e), E(g_X)\cap E_-(e)\big)$\;
        \List{$S',g_{C}',g_{X}'$}\;
        $E(g_C)\gets E(g_C)\setminus\{e\}$ and $E(g_X)\gets E(g_X)\cup \{e\}$\;
    }
}
\end{algorithm}




\begin{theorem}
\label{theo:ebbmc}
    Given a graph $G=(V,E)$, the time and space complexities of \texttt{EBBMC} with {\cheng the} truss-based edge ordering are $O(\delta m +\tau m \cdot 2^{\tau})$ and $O(m+n)$, respectively.  
\end{theorem}

\begin{proof}
    \texttt{EBBMC} first needs $O(\delta m)$ time to obtain the truss-based edge ordering \cite{che2020accelerating}. Then it creates $O(m)$ sub-branches from the initial branch $B=(\emptyset, G, G[\emptyset])$. Since (1) the edge-oriented branching corresponds to a recursive binary partitioning process, which implies that each sub-branch would explore $O(2^{\tau})$ branches in the worst case and (2) the size of a maximal clique can be bounded by $\tau$, the worst-case time complexity of \texttt{EBBMC} is $O(\delta m + \tau m \cdot 2^{\tau})$. 
\end{proof}


\subsection{A Hybrid Branch-and-bound {\cheng Framework}: \texttt{HBBMC}}
\label{subsec:HBBMC}

{\chengc On the one hand,} {\cheng the edge-oriented branching framework with the} truss-based edge ordering can produce much smaller graph {\cheng instances} than those produced by vertex-oriented branching strategy {\cheng since} $\tau<\delta$ \cite{wang2024technical}.
{\chengc On the other hand,}
{\cheng 
the vertex-oriented branching framework can be improved with the pivoting strategy~\cite{tomita2006worst,naude2016refined}.} 
{\chengc Therefore,} 
we propose to integrate the pivoting strategy into the \texttt{EBBMC} framework {\chengc to enjoy the merits of both frameworks}. 
{\chengB Specifically,} we observe that the maximal cliques in $B$ can be enumerated either with \texttt{VBBMC} frameworks (where the branches are produced based on the vertices in $g_C$ according to Eq.~(\ref{eq:vbb-branching})) or with \texttt{EBBMC} framework (where the branches are produced mainly based on the edges in $g_C$ according to Eq.~(\ref{eq:ebb-branching}) and Eq.~(\ref{eq:zbb-branching})). To inherit the benefits from these two frameworks, we first adopt \texttt{EBBMC} {\cheng for the branching at} the initial branch $B=(\emptyset, G, G[\emptyset])$ since the produced candidate graph instance in each sub-branch is bounded by $\tau$; then for each produced sub-branch, {\kaixinC to achieve the best theoretical performance,} we adopt \texttt{VBBMC} with classic pivot selection  \cite{tomita2006worst, eppstein2010listing} to boost the efficiency. 
{\kaixinC Specifically, given a branch $B=(S,g_C,g_X)$, the pivoting strategy always selects a vertex $v^*\in V(g_C)\cup V(g_X)$ such that $v^*$ has the most number of neighbors in $g_C$.}
We call {\chengB this \underline{h}ybrid \underline{b}ranch-and-\underline{b}ound framework for \underline{MC}E} {\chengc as} \texttt{HBBMC}, whose pseudo-code is presented in Algorithm~\ref{alg:HBBMC}. 
{\kaixinC We note that \texttt{HBBMC} is the first algorithm that integrates the edge-oriented branching strategy and the vertex-oriented branching strategy 
for enumeration problem. Besides, with this new hybrid branching strategy, the worst-case time complexity can be theoretically improved when the graph satisfies some condition, which we show in the following theorem as follows.}

\begin{algorithm}[t]
\caption{The hybrid branching-based BB framework for maximal clique enumeration: \texttt{HBBMC}}
\label{alg:HBBMC}
\KwIn{A graph $G=(V,E)$.}
\KwOut{All maximal cliques within $G$.}
\SetKwFunction{List}{\textsf{VBBMC\_Rec}}
\SetKwProg{Fn}{Procedure}{}{}
$\pi_{\tau}(G)\gets$ the truss-based ordering of edges in $G$\;
\For{each edge $e_i\in E(G)$ following $\pi_{\tau}(G)$}{
    Obtain $S_i$, $g_{C_i}'$ and $g_{X_i}'$ (lines 6-7 of Algorithm~\ref{alg:EBBMC})\;
    \List{$S_i,g_{C_i}',g_{X_i}'$}\;
}
Conduct Termination (lines 9-12 of Algorithm~\ref{alg:EBBMC})\;
    
\end{algorithm}

\begin{theorem}
\label{theo:hbbmc}
    Given a graph $G=(V,E)$, the time and space complexities of \texttt{HBBMC} are $O(\delta m + \tau m \cdot 3^{\tau/3})$ and $O(m+n)$, respectively.
\end{theorem}

\begin{proof}
    The running time of \texttt{HBBMC} consists of (1) the time of generating the truss-based edge ordering and obtaining the subgraphs (i.e., $g_C$ and $g_X$) for enumeration, which can be done in $O(\delta m)$ time, where $\delta$ is the degeneracy of the graph, and (2) the time of recursive enumeration procedure, which can be done in $O(\tau m \cdot 3^{\tau/3})$ time. Detailed analysis can be found in the Appendix~\ref{app:proof-theorem2}. 
\end{proof}

\smallskip\noindent\textbf{{\chengc Differences between \texttt{HBBMC} and \texttt{EBBkC}~\cite{wang2024technical}.}}
{\kaixinC 
In specific, our \texttt{HBBMC} algorithm targets the maximal clique enumeration problem and has the worst-case time complexity of $O(\delta m + \tau m \cdot 3^{\tau/3})$ while \texttt{EBBkC} in \cite{wang2024technical} targets the $k$-clique listing problem and has the worst-case time complexity of $O(\delta m + km \cdot (\tau/2)^{k-2})$. Although these two theoretical results both come from the analysis of the recurrence of the time cost $T(n)$, where $n$ is the size of candidate graph, they follow totally different guidelines and have different forms of the recurrence inequalities (Eq.~(\ref{eq:rec}) of this paper v.s. Eq. (6) in \cite{wang2024technical}). To see this, we explain it from the perspective of the recursion tree (i.e., both algorithms are branch-and-bound algorithms and form a recursion tree with each branch being a tree node).
\begin{itemize}[leftmargin=*]
    \item The depth of the recursion tree formed by \texttt{HBBMC} is bounded by $\tau$ while that formed by \texttt{EBBkC} is bounded by $k$ (note that the partial solution of \texttt{HBBMC} (resp. \texttt{EBBkC}) contains at most $\tau$ (resp. $k$) vertices and we add at least one vertex to the partial set at each branching), where $\tau$ can be much larger than $k$ in practice.
    
    \item The width of the recursion tree formed by \texttt{HBBMC} is different from that of \texttt{EBBkC}. Note that for a branch $B=(S,g_C,g_X)$ in \texttt{HBBMC}, it forms $V(g_C)\setminus N(v^*, g_C)$ sub-branches based on the vertex-oriented branching with pivot selection, where $v^*$ is the pivot vertex, while for a branch $B$ in \texttt{EBBkC}, where $B$ is associated with a partial clique $S$ and a candidate graph $g$, it forms at most $V(g)$ sub-branches based on the vertex-oriented branching. 
\end{itemize}

Due to the above differences, the recursion tree of \texttt{HBBMC} would be narrower but deeper while the recursion tree of \texttt{EBBkC} in \cite{wang2024technical} would be wider but shallower. This also contributes to the different number of branches in the recursion tree. {\chengc Specifically}, the number of branches in the tree of \texttt{HBBMC} could be exponential wrt the size of the input (e.g., $O(3^{n/3})$) while that of \texttt{EBBkC} is polynomial (e.g., $O(n^k)$) wrt the size of the input. Therefore, the problems of solving the recurrence inequalities {\chengc are} largely different from each other.

Furthermore, the correctness of \texttt{HBBMC} is based on the edge-oriented BK branching, which is a newly introduced branching strategy that is not considered in any existing work, including \cite{wang2024technical}. 
{\chengc Specifically,}
we refer to two formulated branching steps of ours and \cite{wang2024technical} (Eqs.~(\ref{eq:ebb-branching}) and (\ref{eq:zbb-branching}) of this paper v.s. Eq.~(2) of \cite{wang2024technical}), respectively. The differences are in twofold. First, the data structure $g_X$ must be maintained for maximality checking {\chengc for the} maximal clique enumeration problem such that each maximal clique is enumerated only once but $g_X$ is not necessary for the $k$-clique listing problem. Second, Eq. (3) is solely for our problem, for which we have explained in Section~\ref{subsec:EBBMC}. 

Lastly, the term ``hybrid'' refers to different meanings in our paper and in \cite{wang2024technical}, respectively. In our paper, the ``hybrid'' in \texttt{HBBMC} refers to different enumeration frameworks (e.g., the vertex-oriented one and the edge-oriented one) while \cite{wang2024technical} solely utilizes edge-based branching. The achieved theoretical results in this paper (Theorem~\ref{theo:hbbmc}) rely not only on the edge ordering but also on the pivot selection strategy.}

{\smallskip\noindent\textbf{Remarks}. 
First, compared with \texttt{VBBMC} algorithms, we define the edge density of the graph as $\rho=\frac{m}{n}$. When $\delta \ge \max\{3, \tau + \frac{3\ln\rho}{\ln 3}\}$, the time complexity of \texttt{HBBMC} (i.e., $O(\delta m + \tau m \cdot 3^{\tau/3})$) {is bounded by} the time complexity of the state-of-the-art \texttt{VBBMC} algorithms (i.e., $O(n\delta \cdot 3^{\delta/3})$). The reason is as follows. {First}, $O(\delta m)$ is bounded by $O(n\delta \cdot 3^{\delta/3})$ when $\delta \ge 3$ since (1) $m \le n\delta$ \cite{lick1970k} and (2) $\delta \le 3^{\delta/3}$ holds for all $\delta$ when $\delta\ge 3$. 
{Second}, $O(\tau m \cdot 3^{\tau/3})$ is bounded by $O(n\delta \cdot 3^{\delta/3})$ when $\delta \ge \tau + \frac{3\ln\rho}{\ln 3}$. To see this, 
given $\delta \ge \tau + \frac{3 \ln\rho}{\ln 3}$, 
we have $\tau m \cdot 3^{\tau/3} < \delta \cdot n\rho \cdot 3^{\delta/3 - \frac{\ln\rho}{\ln 3} } = \delta  n \cdot 3^{\delta/3}$, where the inequality derives from the facts $\tau < \delta$ \cite{wang2024technical}, $m=n\rho$ and the given condition, and the equation derives from the fact $3^{\frac{\ln \rho}{\ln 3}}=\rho$. We collect the values of $\delta$, $\tau$ and $\rho$ on 139 real-world graphs from \cite{RealGraphs} and 84 real-world graphs from \cite{RealGraphs2}, finding that the condition $\delta \ge \max\{3, \tau + \frac{3 \ln\rho}{\ln 3}\}$ is satisfied for the majority of the graphs (111 {\chengB out} of 139 and 72 {\chengB out} of 84, respectively). 
Second, there are several different \texttt{VBBMC} algorithms in existing studies, which implies that \texttt{HBBMC} can incorporate \texttt{EBBMC} with different \texttt{VBBMC} implementations. We choose the \texttt{VBBMC} with {\kaixinC the pivot selection strategy used in \cite{tomita2006worst, eppstein2010listing}} as the default setting since it achieves the best theorectical results and the relatively good performance in practice (details {\chengB will be presented} in Section~\ref{sec:exp}). 
}

\section{Early Termination Technique}
\label{sec:early-termination}

Suppose that we are at a branch $B=(S, g_C, g_X)$, where $g_C$ is dense ({\cheng e.g., it is} a clique or nearly a clique) and {\kaixin $g_X$ is an empty graph}, and we aim to enumerate all maximal cliques within $B$. While the pivot strategy in \texttt{VBBMC} and \texttt{HBBMC} has proven to be effective in practice, it still needs to conduct {\cheng many} set intersections within the vertices in $g_C$,
which would be costly. To overcome this issue, we propose an \emph{early-termination} technique, which can directly enumerate the maximal cliques in $B$ with a dense $g_C$ without further {\cheng forming sub-branches on} it. Specifically, we have the following observations. 
If $g_C$ is a clique (a.k.a 1-plex), a 2-plex or a 3-plex, we can directly enumerate all maximal cliques within $g_C$ in nearly optimal time. For the first case, {\kaixin since $g_X$ is an empty graph, we can directly include the vertices in $g_C$ to $S$ to form a maximal clique of $G$. 
For the other two cases, 
we construct the inverse graph $\overline{g}_C$ of $g_C$. 
{An observation is that $\overline{g}_C$} only has several simple paths or cycles. With these additional topological information, we can output all maximal cliques in {\chengc an amount of} time that is proportional to the number of maximal cliques in $g_C$ (details will be discussed in Section~\ref{subsec:2-plex} and Section~\ref{subsec:3-plex}).}

\subsection{Enumerating Maximal Cliques from 2-Plex in Nearly Optimal Time}
\label{subsec:2-plex}

\begin{algorithm}[t]
\caption{Enumerate maximal cliques in a 2-plex}
\label{alg:MC2Plex}
\KwIn{A branch $B=(S,g_C,g_X)$, where $g_C$ is to a 2-plex and $g_X$ is an empty graph.}
\KwOut{All maximal cliques within $B$.}
\SetKwFunction{List}{\textsf{HBBMC-T\_Rec}}
\SetKwProg{Fn}{Procedure}{}{}
Partition $V(g_C)$ into three disjoint sets $F$, $L$ and $R$, such that $L[i]$ disconnects with $R[i]$ in $g_C$ for $1\le i\le |L|$\;
\For{$num=0$ to $2^{|L|}-1$}{
    $S'\gets F$\;
    \For{$i=1$ to $|L|$}{
        \lIf{the $i$-th bit of $num$ is 0}{
            $S'\gets S'\cup \{L[i]\}$
        }
        \lElse {
            $S'\gets S'\cup \{R[i]\}$
        }
    }

        \textbf{Output} a maximal clique $S\cup S'$\;
}
\end{algorithm}

Consider a branch $B=(S, g_C, g_X)$ where $g_C$ is a 2-plex {\kaixin and $g_X$ is an empty graph}. 
We can partition all vertices in $g_C$ into three disjoint sets, namely set $F$, set $L$ and set $R$, each of which induces a clique. This process can be done in two steps. First, we partition $V(g_C)$ into two sets: one containing the vertices that are connected with all the other vertices in $g_C$ (this is $F$) and the other containing the remaining vertices (this is $L\cup R$). For the latter, since each vertex in a 2-plex disconnects with at most 2 vertices (including itself), $L\cup R$ must contain a collection of pairs of vertices $\{u,v\}$ such that $u$ disconnects with $v$ in $g_C$. Thus, we can break each pair of vertices and partition $L\cup R$ into two disjoint sets, where $L$ and $R$ include the first and the second vertices in these pairs, respectively. We note that the {\chengB partitioning} of $L\cup R$ is not unique. Consider an arbitrary one. We observe that the vertex set $F\cup L$ induces a maximal clique in $g_C$. The reason is that if there exists a larger clique $C'$ that is induced by $F\cup L\cup R'$ where $R'\neq \emptyset$, then we must have $R'\subseteq R$. We can find a subset of $L$, denoted by $L'$, such that each vertex in $L'$ disconnects with a vertex in $R'$. In this case, $C'$ is not a clique. Then consider the number of different partitions of $L\cup R$. For each pair of vertices $\{u,v\}$ such that $u$ disconnects with $v$, we can either have (1) $u\in L$ and $v\in R$ or (2) $v\in L$ and $u\in R$. Thus, there are $2^{|L|}$ different partitions of $L\cup R$. 
That is, we have $2^{|L|}$ maximal cliques in $g_C$. We summarize the above procedure in Algorithm~\ref{alg:MC2Plex}.

\begin{theorem}
\label{theo:complex-within-2plex}
    Given a branch $B=(S,g_C,g_X)$ with $g_C$ being a 2-plex and $g_X$ being an empty graph, Algorithm~\ref{alg:MC2Plex} enumerates all maximal cliques within $B$ in $O(E(g_C) + \omega(g_C) \cdot c(g_C))$ time, where $\omega(g_C)$ is the size of the largest maximal clique within $B$ and $c(g_C)$ is the number of maximal cliques in $g_C$.
\end{theorem}

\begin{proof}
    First, Algorithm~\ref{alg:MC2Plex} takes $O(E(g_C))$ {\cheng time} for partitioning $V(g_C)$ into three sets (line 1). In each round of lines 2-7,  the algorithm guarantees exactly one maximal clique of $g_C$ to be found. 
\end{proof}


{\kaixin We note that the above procedure achieves a \emph{nearly optimal} time complexity since $O(\omega(g_C) \cdot c(g_C))$ is the optimal time  and it only takes $O(E(g_C))$ {\chengB extra} time.}

\subsection{{\cheng Enumerating} Maximal Cliques from 3-Plex in Nearly Optimal Time}
\label{subsec:3-plex}

Consider a branch $B=(S,g_C, g_X)$ where $g_C$ is 3-plex and $g_X$ is an empty graph. 
We construct the inverse graph $\overline{g}_C$ of $g_C$. We denote all isolated vertices in $\overline{g}_C$ by $F$ since they connect with all other vertices in $g_C$. Then for each vertex $v\in V(\overline{g}_C)\setminus F$, it belongs to a connected component which can be categorized into two classes: simple paths and simple cycles. The reason is that (1) the maximum degree of a vertex in $\overline{g}_C$ is no greater than two and (2) all connected components of a graph with a maximum degree less than or equal to two consist exclusively of simple paths or simple cycles \cite{west2001introduction}.
We elaborate on the details of enumerating maximal cliques within these simple paths and simple cycles as follows.

\noindent\textbf{Category 1 - Paths.}
For a path $p=\{v_1, v_2, \cdots, v_{|p|}\}$ in $\overline{g}_C$ ($|p|\ge 2$), the vertices are in the order where $(v_i, v_{i+1})\in E(\overline{g}_C)$. 
Correspondingly, it implies that $v_1$ and $v_{|p|}$ disconnect with only one vertex in $g_C$ (i.e., $v_2$ and $v_{|p|-1}$, respectively), and for any other vertices $v_i$ (if any), it disconnects with two vertices $v_{i-1}$ and $v_{i+1}$ in $g_C$ for $2\le i\le |p|-1$. Next, we present how we can enumerate all maximal cliques in $g_C$ within the vertices of a path $p$ of $\overline{g}_C$. The procedure can be summarized in three steps.

\begin{itemize}[leftmargin=*]
    \item \textbf{Step 1 - {\chengB Initialization}.}  Let $S_p$ be a subset of $V(p)$ that induces a partial clique of $g_C$. We initialize $S_p=\{v_1\}$ or $S_p=\{v_2\}$. 
    \item \textbf{Step 2 - {\chengB Expansion}.} Let $v_i$ be the last vertex of $S_p$. Then $S_p$ can be either expanded (1) by adding $v_{i+2}$ into $S_p$ (i.e., $S_p = S_p\cup \{v_{i+2}\}$) or (2) by adding $v_{i+3}$ into $S_p$ (i.e., $S_p = S_p\cup \{v_{i+3}\}$). Repeat Step 2 until no vertex can be {\cheng used to expand} $S_p$, i.e., $i+2>|p|$, then go to Step 3. 
    \item \textbf{Step 3 - {\chengB Termination}.} Report $S_p$ as a maximal clique induced within the vertices of the path $p$.
\end{itemize}

\begin{algorithm}[t]
\caption{Function \textsf{Enum\_from\_Path($p$)}}
\label{alg:enum_path}
\KwIn{A path $p=\{v_1, v_2, \cdots, v_{|p|}\}$ in $\overline{g}_C$.}
\KwOut{All maximal cliques of $g_C$ induced with the vertices in $p$.}
\SetKwFunction{List}{\textsf{Enum\_Rec}}
\SetKwProg{Fn}{Procedure}{}{}
$\mathbb{S}\gets \emptyset$\;
\List{$\mathbb{S}, \{v_1\}, p$}; \List{$\mathbb{S}, \{v_2\}, p$}\;
\Return $\mathbb{S}$\;
\Fn{\List{$\mathbb{S}, S_p, p$}}{
    $v_i\gets$ the last vertex in $S_p$\;
    \If{$i+2>|p|$}{
        $\mathbb{S}\gets \mathbb{S} \cup \{S_p\}$;
        \Return\;
    }
    \List{$\mathbb{S}, S_p\cup\{v_{i+2}\}, p$}\;
    \lIf{$i+3\le|p|$}{
        \List{$\mathbb{S}, S_p\cup\{v_{i+3}\}, p$}
    }
}
\end{algorithm}

\begin{algorithm}[t]
\caption{Function \textsf{Enum\_from\_Cycle($c$)}}
\label{alg:enum_cycle}
\KwIn{A cycle $c=\{v_1, v_2, \cdots, v_{|c|}\}$ in $\overline{g}_C$.}
\KwOut{All maximal cliques of $g_C$ induced with the vertices in $c$.}
\SetKwFunction{List}{\textsf{Enum\_Rec}}
\SetKwProg{Fn}{Procedure}{}{}
\lIf{$|c|=3$}{\Return $\{\{v_1\},\{v_2\},\{v_3\}\}$}
\lIf{$|c|=4$}{\Return $\{\{v_1,v_3\},\{v_2,v_4\}\}$}
\lIf{$|c|=5$}{\Return $\{\{v_1,v_3\},\{v_1, v_4\},\{v_2,v_4\},\{v_2,v_5\},\{v_3,v_5\}\}$}
\If{$|c|\ge 6$} {
$\mathbb{S}\gets\emptyset$\;
\List{$\mathbb{S}, \{v_1\}, \{v_1, v_2, \cdots, v_{|c|-1}\}$}\;
\List{$\mathbb{S}, \{v_2\}, \{v_2, v_3, \cdots, v_{|c|}\}$}\;
\List{$\mathbb{S}, \{v_{|c|}, v_3\}, \{v_3, v_4, \cdots, v_{|c|-2}\}$}\;
\Return $\mathbb{S}$\;
}
\end{algorithm}

\noindent\textbf{{\chengB Correctness}.}
{First}, when both $v_1$ and $v_2$ are not in $S_p$, $S_p$ cannot be maximal since we can add $v_1$ into $S_p$ to form a larger clique. {Second}, it is easy to see that when $v_{i}\in S_p$, $v_{i+1}$ cannot be in $S_p$ since they are disconnected with each other in $g_C$. Thus, one possible choice is to add $v_{i+2}$ to $S_p$. Besides, there exists another choice that we also exclude $v_{i+2}$ from $S_p$ but include $v_{i+3}$ into $S_p$. In this case, adding $v_{i+1}$ or $v_{i+2}$ to $S_p$ cannot form a larger clique since $v_i$ disconnects with $v_{i+1}$ in $g_C$ and $v_{i+2}$ disconnects with $v_{i+3}$ in $g_C$. For other expansion that does not include $v_{i+2}$ and $v_{i+3}$ {\cheng inside}, $S_p$ cannot be maximal since we can add $v_{i+2}$ into $S_p$ to form a larger clique.  
An illustration of enumerating maximal cliques in a given path $p$ is shown in Figure~\ref{fig:path-1} and Figure~\ref{fig:path-2}. The pseudo-codes of function \textsf{Enum\_from\_Path($p$)} is shown in Algorithm~\ref{alg:enum_path}. 
\noindent\textbf{Category 2 - Cycles.}
For a cycle $c = \{v_1, v_2, \cdots, v_{|c|}\}$ in $\overline{g}_C$ ($|c|\ge 3$), the vertices are in the order where $(v_i, v_{i+1})\in E(\overline{g}_C)$. 
When $|c|=3$, $|c|=4$ and $|c|=5$, the maximal cliques are $\{\{v_1\}, \{v_2\}, \{v_3\}\}$, $\{\{v_1, v_3\}, \{v_2, v_4\}\}$ and $\{\{v_1,v_3\},\{v_1, v_4\},\{v_2,v_4\},\{v_2,v_5\},\{v_3,v_5\}\}$, respectively. 
For $|c|\ge 6$, there are three possible cases, each of which can be reduced to enumerate maximal cliques given a simple path in $\overline{g}_C$ as described in Category 1. Let $S_c$ be a subset of $V(c)$ that induces a maximal clique of $g_C$.

\begin{itemize}[leftmargin=*]
    \item \textbf{Case 1 - $v_1\in S_c$.} Consider \emph{path} $p = \{v_1, v_2, \cdots, v_{|c|-1}\}$. All maximal cliques within $c$ can be enumerated by following three steps in Category 1 with initialization $S_p=\{v_1\}$. 

    \item \textbf{Case 2 - $v_2\in S_c$.} Consider \emph{path} $p=\{v_2, v_3, \cdots, v_{|c|}\}$. All maximal cliques within $c$ can be enumerated by following three steps in Category 1 with initialization $S_p=\{v_2\}$. 

    \item \textbf{Case 3 - $v_1, v_2\notin S_c$.} Consider \emph{path} $p=\{v_3, \cdots, v_{|c|-2}\}$. All maximal cliques within $c$ can be enumerated by following three steps in Category 1 with initialization $S_p=\{v_{|c|}, v_3\}$. 
\end{itemize}

\begin{figure}[t]
\centering
\subfigure[Expand $S_p$ with $v_{i+2}$.]{
\label{fig:path-1}
\includegraphics[width=0.23\textwidth]{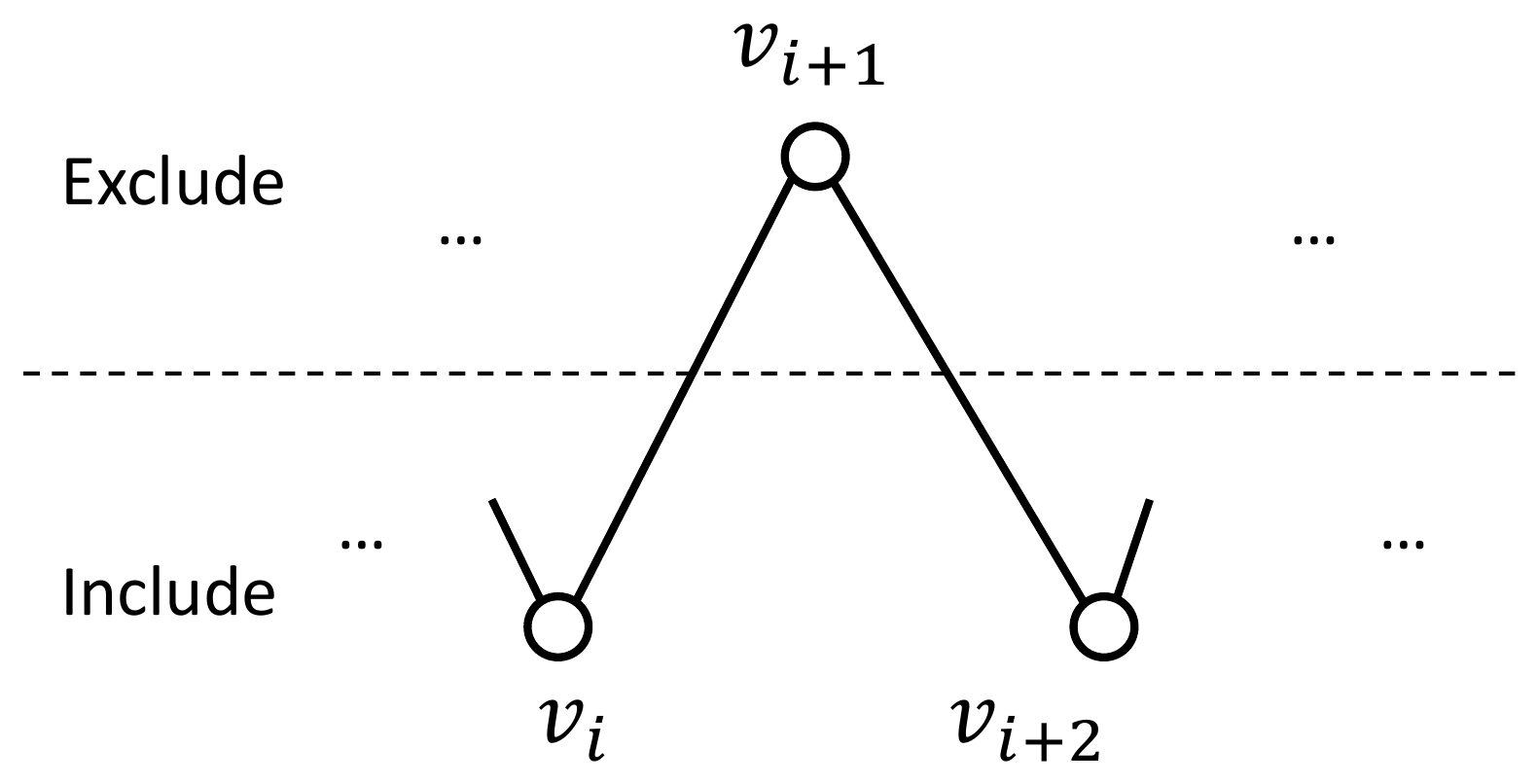}}
\subfigure[Expand $S_p$ with $v_{i+3}$.]{
\label{fig:path-2}
\includegraphics[width=0.23\textwidth]{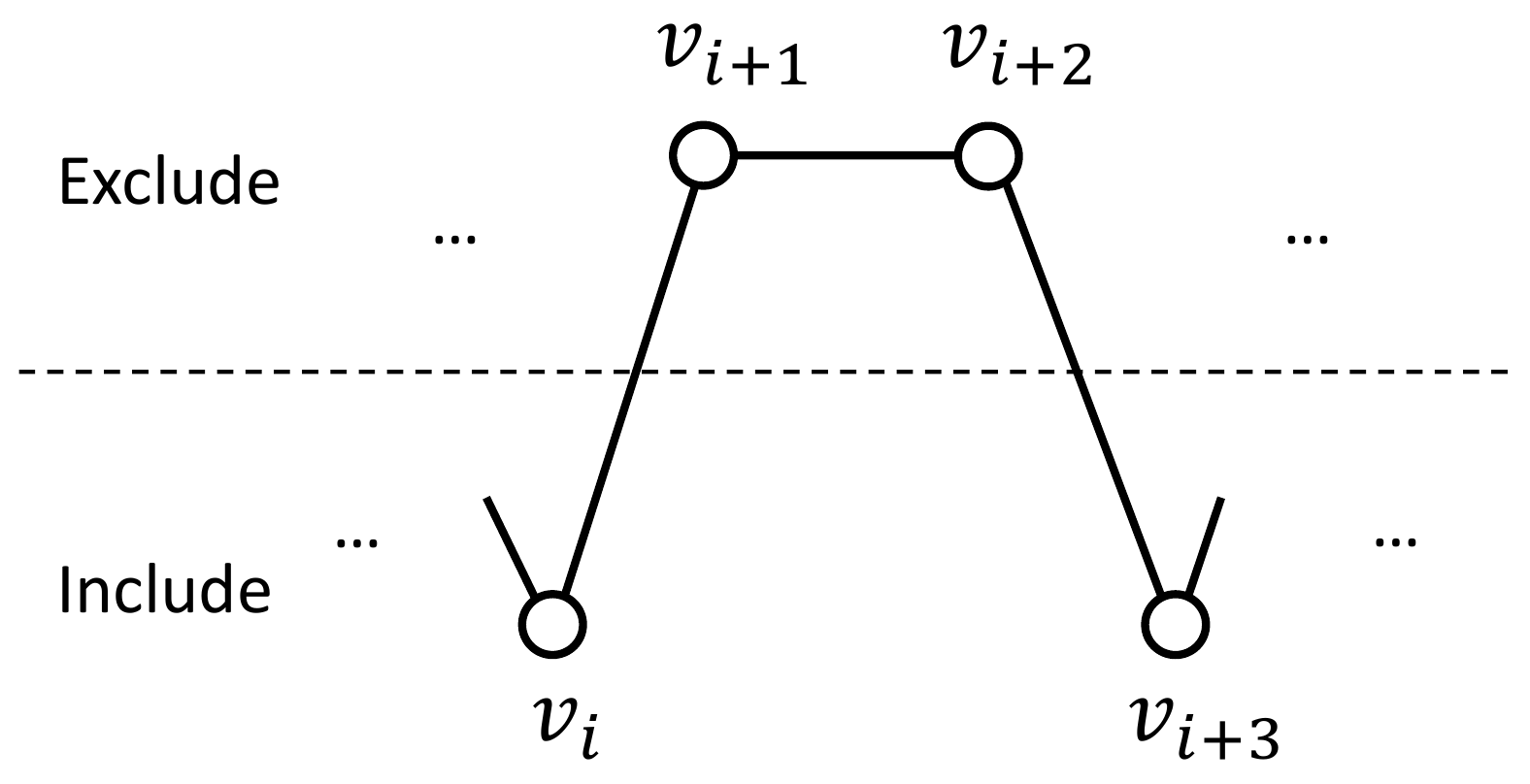}}
\subfigure[Case 1 \& 2: $v_1\in S_c$ or $v_2\in S_c$.]{
\label{fig:cycle-1}
\includegraphics[width=0.225\textwidth]{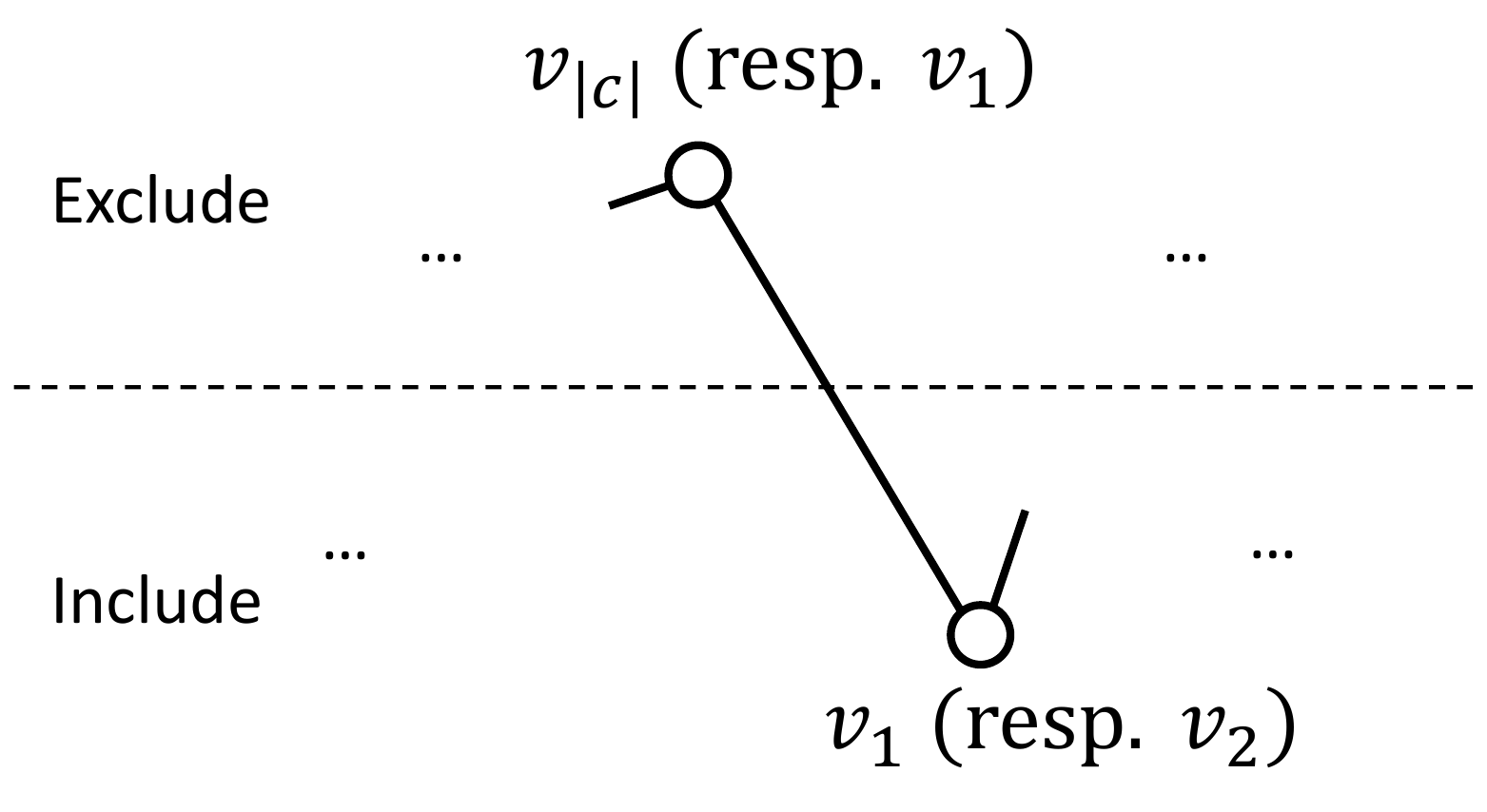}}
\subfigure[Case 3: $v_1, v_2\notin S_c$.]{
\label{fig:cycle-2}
\includegraphics[width=0.23\textwidth]{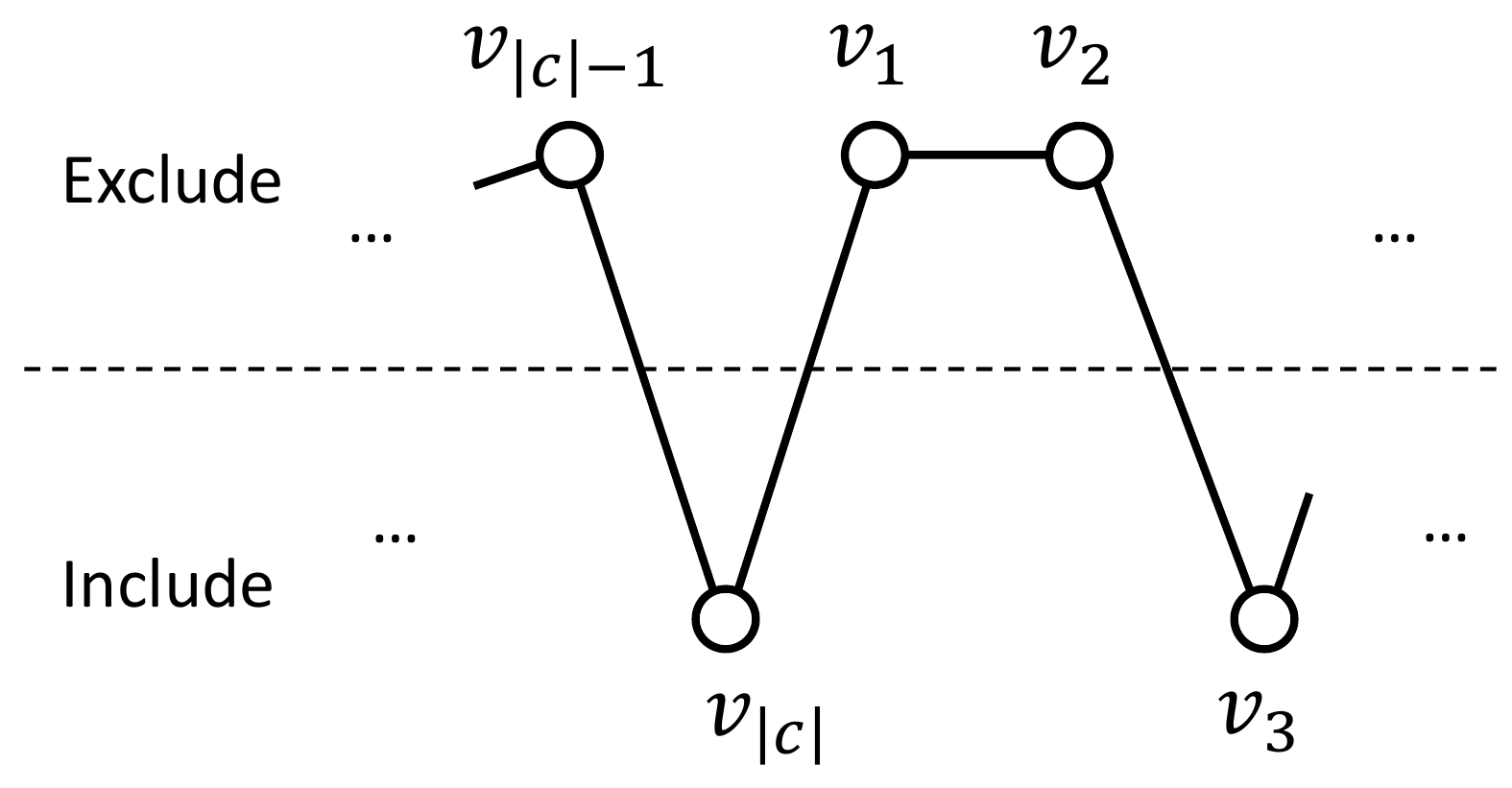}}
\caption{ (a) and (b) show two possible cases to expand a partial maximal clique $S_p$ of a given path $p$ when $v_i\in S_p$ is the last vertex. 
(c) and (d) show three possible cases to enumerate all maximal cliques given a cycle $c$. Note that ``Include'' (resp. ``Exclude'') means that the vertices that are included in (resp. excluded from) a maximal clique. }
\label{fig:cases}
\end{figure}


\begin{algorithm}[t]
\caption{{Enumerate} maximal cliques in a 3-plex}
\label{alg:MC3Plex}
\KwIn{A branch $B=(S,g_C,g_X)$ where $g_C$ is to a 3-plex and $g_X$ is an empty graph.}
\KwOut{All maximal cliques within $B$.}
\SetKwFunction{ListPath}{\textsf{Enum\_from\_Path}}
\SetKwFunction{ListCycle}{\textsf{Enum\_from\_Cycle}}
\SetKwProg{Fn}{Procedure}{}{}
Construct the inverse graph $\overline{g}_C$ of $g_C$\;
$F, \mathcal{P}, \mathcal{C}\gets$ set of the isolated vertices, paths and cycles in $\overline{g}_C$\;
\lFor{each path $p_i\in \mathcal{P}$}{
$\mathbb{S}_{p_i}\gets$\ListPath{$p_i$}
}
\lFor{each cycle $c_i\in \mathcal{C}$}{
$\mathbb{S}_{c_i}\gets$\ListCycle{$c_i$}
}

\For{each combination $(S_{p_1},S_{p_2},\cdots,S_{p_{|\mathcal{P}|}})$ where $S_{p_i}\in \mathbb{S}_{p_i}$}{
\For{each combination $(S_{c_1},S_{c_2},\cdots,S_{c_{|\mathcal{C}|}})$ where $S_{c_j}\in \mathbb{S}_{c_j}$}{
    $S'\gets F\cup \Big(\bigcup_{i=1}^{|\mathcal{P}|} S_{p_i}\Big) \cup \Big(\bigcup_{j=1}^{|\mathcal{C}|} S_{c_j}\Big)$\;
        \textbf{Output} a maximal clique $S\cup S'$\;
}
}
\end{algorithm}

\noindent\textbf{{\chengB Correctness}.}
For Case 1, when $v_1\in S_c$, $v_{|c|}$ must not be in $S_c$ since they are connected in $\overline{g}_C$. It corresponds to finding maximal cliques from the vertices in $\{v_1, v_2, \cdots, v_{|c|-1}\}$ with $v_1$ required to be inside. For Case 2, when $v_2\in S_c$, $v_1$ must not be in $S_c$. It corresponds to finding maximal cliques from the vertices in $\{v_2, \cdots, v_{|c|}\}$ with $v_2$ required to be inside. For Case 3, when $v_1,v_2\notin S_c$, to ensure the maximality of $S_c$, $v_3$ and $v_{|c|}$ must be in $S_c$ since otherwise, we can add $v_1$ or $v_2$ to $S_c$ to form a larger clique. As $v_3$ and $v_{|c|}$ are both in $S_c$, $v_{|c|-1}$ must not be in $S_c$. It corresponds to finding maximal cliques from the vertices in $\{v_3, \cdots, v_{|c|-2}\}$ with $v_3$ required to be inside. 
An illustration of enumerating maximal cliques in a given cycle $c$ is shown in Figure~\ref{fig:cycle-1} and Figure~\ref{fig:cycle-2}. The pseudo-codes of function \textsf{Enum\_from\_Cycle($c$)} is shown in Algorithm~\ref{alg:enum_cycle}.


Once the maximal cliques within each cycle and path are enumerated, we combine them together to form the maximal cliques in $g_C$. We summarize the procedure in lines {\yui 5-8} of Algorithm~\ref{alg:MC3Plex}. 

\begin{theorem}
\label{theo:complex-within-3plex}
    Given a branch $B=(S,g_C,g_X)$ with $g_C$ being a 3-plex and $g_X$ {\cheng being} an empty graph, Algorithm~\ref{alg:MC3Plex} enumerates all maximal cliques within $B$ in $O(E(g_C) + \omega(g_C) \cdot c(g_C))$ time, where $\omega(g_C)$ is the size of the largest maximal clique within $B$ and $c(g_C)$ is the number of maximal cliques in $g_C$.
\end{theorem}

\begin{proof}
    First, Algorithm~\ref{alg:MC3Plex} takes $O(E(g_C))$ {\cheng time} for partitioning $V(g_C)$ into three sets (line 1). In each round of lines 7-10,  the algorithm guarantees exactly one maximal clique of $g_C$ to be found.
\end{proof}


{\kaixin Similarly, Algorithm~\ref{alg:MC3Plex} also achieves a \emph{nearly optimal} time complexity since $O(\omega(g_C) \cdot c(g_C))$ is the optimal time  and Algorithm~\ref{alg:MC3Plex} only takes $O(E(g_C))$ {\chengB extra} time.}

\begin{table*}[ht]
\vspace{-3mm}
    \begin{minipage}{.46\textwidth}
    \scriptsize
    \centering
    \caption{Dataset Statistics.}
    \vspace{-2mm}
    \label{tab:data}
    \begin{tabular}{r|r|rr|rrr}
    \hline
        \textbf{Graph (Name)} & {Category} & $|V|$ & $|E|$ & $\delta$ & $\tau$ & $\rho$ \\
    \hline
        \textsf{nasasrb (NA)} & {Social Network} & 54,870 & 1,311,227 & 35 & 22 &  23.9 \\ 
        \textsf{fbwosn (FB)} & {Social Network} & 63,731 & 817,090 & 52 & 35 &  12.8 \\ 
        \textsf{websk (WE)} & {Web Graph} & 121,422 & 334,419 & 81 & 80 &  2.8 \\ 
        \textsf{wikitrust (WK)} & {Web Graph} & 138,587 & 715,883  & 64 & 31 & 5.2 \\ 
        \textsf{shipsec5 (SH)} & {Social Network} & 179,104 & 2,200,076  & 29 & 22 &  12.3 \\ 
        \textsf{stanford (ST)} & {Social Network} & 281,904 & 1,992,636 & 86 & 61 & 7.1 \\ 
        \textsf{dblp (DB)} & {Collaboration} & 317,080 & 1,049,866 & 113 & 112 & 3.3 \\ 
        \textsf{dielfilter (DE)} & {Other} & 420,408 & 16,232,900 & 56 & 43 & 38.6 \\ 
        \textsf{digg (DG)} & {Social Network} & 770,799 & 5,907,132 & 236 & 72 & 7.7  \\ 
        \textsf{youtube (YO)} & {Social Network} &1,134,890 & 2,987,624 & 49 & 18 & 2.6 \\ 
        \textsf{pokec (PO)} & {Social Network} & 1,632,803 & 22,301,964 & 47 & 27 & 13.7 \\ 
        \textsf{skitter (SK)} & {Web Graph} & 1,696,415 & 11,095,298 & 111 & 67 & 6.5 \\ 
        \textsf{wikicn (CN)} & {Web Graph} & 1,930,270 & 8,956,902 & 127 & 31 & 4.6 \\ 
        \textsf{baidu (BA)} & {Web Graph} & 2,140,198 & 17,014,946 & 82 & 29 & 8.0 \\ 
        \textsf{orkut (OR)} & {Social Network} & 2,997,166 & 106,349,209 & 253 & 74 & 35.5 \\ 
        \textsf{socfba (SO)} & {Social Network} & 3,097,165 &  23,667,394 &  74 & 29 &  7.6 \\ 
    \hline
    \end{tabular}
    \end{minipage}
    \hspace{15mm}
\begin{minipage}{.46\textwidth}
\scriptsize
    \centering
    \caption{Comparison with baselines (unit: second).}
\vspace{-2mm}
    \label{tab:real}
    \begin{tabular}{c|rrrrr}
    \hline
         & \texttt{HBBMC++} & \texttt{RRef} & \texttt{RDegen} & \texttt{RRcd} & \texttt{RFac} \\ \hline
\textsf{NA} & \textbf{0.33} & 0.58 & 0.48 & \underline{0.46} & 0.61 \\
\textsf{FB} & \textbf{1.10} & 1.78 & 1.67 & \underline{1.24} & 1.70 \\
\textsf{WE} & \textbf{0.02} & 0.11 & \underline{0.08} & 0.12 & 0.17 \\
\textsf{WK} & \textbf{0.57} & 1.12 & 1.03 & \underline{1.01} & 1.68 \\
\textsf{SH} & \textbf{0.45} & 1.05 & 0.98 & \underline{0.78} & 1.15 \\
\textsf{ST} & \textbf{1.26} & 2.15 & 1.70 & \underline{1.67} & 5.07 \\
\textsf{DB} & \textbf{0.16} & 0.53 & \underline{0.47} & 0.49 & 0.83 \\
\textsf{DE} & \textbf{3.82} & 8.29 & 7.47 & \underline{5.76} & 9.91 \\ 
\textsf{DG} & \textbf{239.58} & 1441.22 & \underline{1046.40} & 1518.36 & 1603.08 \\
\textsf{YO} & \textbf{1.47} & 2.85 & 2.32 & \underline{2.19} & 6.45  \\
\textsf{PO} & \textbf{19.31} & 32.47 & \underline{25.96} & 26.38 & 31.66 \\
\textsf{SK} & \textbf{25.15} & 65.27 & 47.11 & \underline{44.90} & 71.96 \\
\textsf{CN} & \textbf{6.03} & 14.07 & \underline{11.18} & 12.65 & 20.37 \\
\textsf{BA} & \textbf{13.81} & 28.67 & 22.61 & \underline{20.59} & 36.51 \\
\textsf{OR} & \textbf{884.20} & 2297.57 & \underline{2200.54} & 2410.93 & 2749.32 \\
\textsf{SO} & \textbf{21.12} & 40.58 & 39.61 & \underline{37.44} & 42.91 \\
\hline
    \end{tabular}
\end{minipage}
\vspace{-4mm}
\end{table*}

\smallskip
\noindent\textbf{Remarks.}
\underline{First}, the early-termination technique is orthogonal to all BB frameworks, including existing \texttt{VBBMC} algorithms and the \texttt{HBBMC} algorithm. With the early-termination technique, all BB algorithms retain the same time complexities provided before but run significantly faster in practice, as verified in experiments. Additionally, the strategy involves setting a small threshold $t$ to apply early-termination to different densities of candidate graphs (i.e., $t$-plexes with varying $t$). We vary $t\in\{1,2,3\}$ {\chengB in experiments} and the results indicate that $t=3$ achieves higher speedup over other choices.
%
\underline{Second}, the overhead of checking the early-termination condition is negligible. In the \texttt{HBBMC} framework, the pivoting strategy prunes branches that do not contain any maximal cliques. For a branch $B = (S, g_C, g_X)$, we calculate $d(v, g_C)$, the number of neighbors in $g_C$ for each $v \in V(g_C) \cup V(g_X)$, and select the vertex $v^*$ with the maximum $d(\cdot, g_C)$ as the pivot. The condition for $g_C$ being a 2-plex (or 3-plex) can be checked simultaneously with pivot selection: $g_C$ is a 2-plex (or 3-plex) if the minimum $d(v, g_C)$ for any $v \in V(g_C)$ is at least $|V(g_C)| - 2$ (or $|V(g_C)| - 3$). Let $u_m$ be the vertex with the minimum $d(\cdot, g_C)$. Since $d(u, g_C) \ge d(u_m, g_C) \ge |V(g_C)| - 2$ (or $|V(g_C)| - 3$) for all $u \in V(g_C)$, $g_C$ is confirmed as a 2-plex (or 3-plex). This check only adds $O(|V(g_C)|)$ time per branch, which we consider acceptable.
{\kaixinC \underline{Third}, while the high-level idea of early-termination technique exists for enumeration problem (e.g., \cite{wang2024technical, dai2023hereditary}), our design is different from the existing ones and specific to MCE problem. 
Specifically, for $k$-clique listing problem \cite{wang2024technical}, the design follows a combinatorial manner to find all vertices that are connected with each other; while for MCE problem, our design iteratively adds all possible vertices until no vertex can be added (i.e., adding any vertex {\chengc would} induce a non-clique). For maximal biclique enumeration problem \cite{dai2023hereditary}, the termination is applied when candidate subgraphs have no edges inside; while for MCE problem, the termination is applied for dense candidate subgraphs.
\underline{Fourth}, our design is totally different from simply converting the original recursive procedure into an iterative loop. The reason is as follows. On the one hand, it is hard to convert the original recursive process into a loop since it needs to manually maintain a stack of various variables for conducting the pivot selection and pruning techniques. Thus, existing algorithms including BK and its variants are all implemented in a recursive way. On the other hand, our introduced method takes the advantages of the topological information of the $t$-plexes (with different values of $t$), which can be naturally implemented in an iterative way as shown in Algorithm~\ref{alg:MC2Plex} and Algorithm~\ref{alg:MC3Plex}. This should also be regarded as one of potential benefits of our method. We also provide examples of enumerating maximal cliques from 2-plex and 3-plex in the Appendix~\ref{app:example}.}

\section{Experiments}
\label{sec:exp}

\subsection{Experimental Setup}
\label{subsec:setup}

\noindent\textbf{Datasets.}
We use both real and synthetic datasets. 
The real datasets can be obtained from an open-source network repository \cite{rossi2015network}. For each graph, we follow existing studies \cite{deng2023accelerating, li2019fast} by ignoring directions, weights, and self-loops (if any). 
Our approach is naturally extendable to directed or weighted graphs. By first extracting all maximal cliques without considering direction or weight, we can subsequently filter the cliques to include only those that satisfy user-defined directional or weighted conditions.
We collect the graph statistics and report them in Table~\ref{tab:data}, where $\rho$ represents the edge density of a graph $G=(V,E)$ {\chengB (i.e., $\rho = |E|/|V|$)}, $\delta$ is the graph degeneracy and $\tau$ is the truss related number whose definition can be found in \cite{wang2024technical}. 
The synthetic datasets are generated based on two random graph generators, namely the Erd\"{o}s-R\'{e}yni (ER) graph model \cite{erdds1959random} and the Barab\'{a}si–Albert (BA) graph model \cite{albert2002statistical}. We provide the detailed descriptions of the synthetic datasets and the corresponding results in Appendix~\ref{app:syn}.

\begin{table*}
\vspace{-5mm}
    \begin{minipage}{.48\textwidth}
\scriptsize
    \centering
    \caption{Ablation studies and effects of different implementations of the hybrid frameworks (unit: second).}
    \vspace{-2mm}
    \label{tab:ablation}
    \begin{tabular}{c|rrr|rrr}
    \hline
      & \texttt{HBBMC++} & \texttt{HBBMC+} & \texttt{RDegen} & \texttt{Ref++} & \texttt{Rcd++} & \texttt{Fac++} \\ \hline
\textsf{NA} & \textbf{0.33} & 0.42 & 0.48 & 0.40 & \underline{0.38} & 0.42 \\
\textsf{FB} & \underline{1.10} & 1.40 & 1.67 & 1.17 & \textbf{0.99} & 1.20 \\
\textsf{WE} & \textbf{0.02} & 0.06 & 0.08 & \underline{0.04} & 0.05 & 0.06 \\
\textsf{WK} & \textbf{0.57} & 0.78 & 1.03 & 0.68 & \underline{0.63} & 0.94  \\
\textsf{SH} & \underline{0.45} & 0.88 & 0.98 & 0.48 & \textbf{0.43} & 0.53  \\
\textsf{ST} & \textbf{1.26} & \underline{1.45} & 1.70 & 1.60 & 1.49 & 3.74  \\
\textsf{DB} & \textbf{0.16} & 0.38 & 0.47 & \underline{0.18} & 0.20 & 0.29  \\
\textsf{DE} & \underline{3.82} & 5.53 & 7.47 & 4.23 & \textbf{3.53} & 5.07  \\
\textsf{DG} & \textbf{239.58} & 521.98 & 1046.40 & 426.28 & \underline{363.25} & 412.58  \\
\textsf{YO} & \textbf{1.47} & 1.92 & 2.32 & 1.80 & \underline{1.66} & 4.08 \\
\textsf{PO} & \textbf{19.31} & \underline{22.33} & 25.96 & 24.15 & 23.54 & 23.55 \\
\textsf{SK} & \textbf{25.15} & 40.81 & 47.11 &  34.85 & \underline{28.78} & 54.45 \\
\textsf{CN} & \textbf{6.03} & \underline{7.50} & 11.18 & 7.59 & 8.19 & 10.88 \\
\textsf{BA} & \textbf{13.81} & 18.73 & 22.61 & 17.51 & \underline{15.09} & 16.89          \\
\textsf{OR} & \textbf{884.20} & 1433.02 & 2200.54 & \underline{923.39} & 1162.74 & 1104.95 \\
\textsf{SO} & \textbf{21.12} & 32.16 & 39.61 & \underline{21.63} & 23.95 & 22.88 \\ \hline
    \end{tabular}
    \end{minipage}
    \hspace{8mm}
    \begin{minipage}{.48\textwidth}
\scriptsize
    \centering
    \caption{Effects of hybrid frameworks, varying the time changing from edge-oriented to vertex-oriented branching strategy.}
    \vspace{-2mm}
    \label{tab:hybrid}
\begin{tabular}{c|rr|rr|rr}
\hline
\multirow{2}{*}{} & \multicolumn{2}{c|}{$d=1$} & \multicolumn{2}{c|}{$d=2$} & \multicolumn{2}{c}{$d=3$}  \\\cline{2-7}
 & Time (s) & \#Calls  & Time (s) & \#Calls & Time (s) & \#Calls \\ \hline
\textsf{NA} & \textbf{0.33} & 365K & 0.99 & 1.57M & 4.99 & 13.3M  \\
\textsf{FB} & \textbf{1.10} & 2.15M & 1.46 & 3.47M & 2.45 & 6.82M \\
\textsf{WE} & \textbf{0.02} & 205K & 0.11 & 467K & 1.29 & 1.45M  \\
\textsf{WK} & \textbf{0.57} & 1.76M & 1.04 & 2.91M & 2.35 & 5.83M  \\
\textsf{SH} & \textbf{0.45} & 1.57M & 0.72 & 3.27M & 1.91 & 10.6M \\
\textsf{ST} & \textbf{1.26} & 1.69M & 1.83 & 3.56M & 11.12 & 14.7M \\
\textsf{DB} & \textbf{0.16} & 537K & 0.27 & 1.43M & 3.05 & 3.61M \\
\textsf{DE} & \textbf{3.82} & 1.29M & 33.28 & 17.4M & 313.02 & 279.1M  \\
\textsf{DG} & \textbf{239.58} & 1.54B & 583.76 & 1.89B & 798.05 & 2.07B  \\
\textsf{YO} & \textbf{1.47} & 3.97M & 1.58 & 6.25M & 1.75 & 8.24M  \\
\textsf{PO} & \textbf{19.31} & 27.9M & 21.48 & 39.0M & 25.48 & 65.1M \\
\textsf{SK} & \textbf{25.15} & 53.8M & 30.86 & 76.8M & 59.09 & 104.5M \\
\textsf{CN} & \textbf{6.03} & 16.6M & 13.57 & 24.9M & 16.57 & 39.8M \\
\textsf{BA} & \textbf{13.81} & 25.1M & 25.18 & 35.4M & 26.43 & 53.5M \\
\textsf{OR} & \textbf{884.20} & 5.58B & 1391.90 & 6.11B & 1829.41 & 6.70B \\
\textsf{SO} & \textbf{21.12} & 42.5M & 28.51 & 61.3M & 38.23 & 108.8M \\ \hline
\end{tabular}
\end{minipage}
\vspace{-5mm}
\end{table*}

\smallskip
\noindent\textbf{Baselines, Metrics and Settings.} 
We choose \texttt{HBBMC} as the default branching-based BB framework for comparison. Besides, we apply two techniques to enhance the performance of \texttt{HBBMC}. One is the early-termination technique (\texttt{ET} for short) as introduced in Section~\ref{sec:early-termination}. We set the the parameter $t=3$ by default {\chengB for the early-termination technique.}
The other is the graph reduction technique (\texttt{GR} for short) proposed in \cite{deng2023accelerating}, which is to eliminate the branches that are created {\cheng based on} some small-degree vertices and report the maximal cliques that {\cheng involve} these vertices directly during the recursive procedure. Note that \texttt{GR} is also orthogonal to frameworks and can be easily adopted to our algorithm. We denote our algorithm with the above two techniques by \texttt{HBBMC++}. We compare our algorithm \texttt{HBBMC++} with four state-of-the-art algorithms proposed in \cite{deng2023accelerating}, namely \texttt{RRef} , \texttt{RDegen}, \texttt{RRcd} and \texttt{RFac} in terms of running time. Specifically, each of these four algorithms can be viewed as the original \texttt{VBBMC} algorithms, i.e., \texttt{BK\_Ref} \cite{naude2016refined}, \texttt{BK\_Degen} \cite{eppstein2010listing}, \texttt{BK\_Rcd} \cite{li2019fast}, \texttt{BK\_Fac} \cite{jin2022fast}, with {\cheng the} graph reduction technique. 
The running time of all algorithms reported in our paper includes the time costs of (1) generating the orderings of vertices or edges (if any) and (2) enumerating maximal cliques via recursive procedures. The source codes of all algorithms are written in C++ and the experiments are conducted on a Linux machine with a 2.10GHz Intel CPU and 128GB memory. 
{\kaixinC The codes are available at \url{https://github.com/wangkaixin219/HBBMC}.}

\subsection{Experimental Results on Real Datasets}
\label{subsec:real-result}

\noindent
\textbf{(1) Comparison among algorithms.}
Table~\ref{tab:real} shows the results of the running time of different algorithms for maximal clique enumeration on all datasets. We observe that our algorithm \texttt{HBBMC++} runs faster than all baselines on all datasets. For most of the datasets that satisfy the condition $\delta \ge \tau+\frac{3}{\ln 3}\ln\rho$, the results are consistent with our theoretical analysis that the time complexity of \texttt{HBBMC++} is better than those of the baseline algorithms. For example, on \textsf{DG}, \textsf{CN} and \textsf{OR}, the values of $\tau$ are far smaller than the values of $\delta$. As a result, \texttt{HBBMC++} can run up to 4.38x faster than the state-of-the-art algorithms. For the others (e.g., \textsf{WE} and \textsf{DB}), the speedups are also significant, i.e., \texttt{HBBMC++} achieves 4.0x and 2.94x speedups on \textsf{WE} and \textsf{DB}, respectively. There are two possible reasons. First, while the \texttt{HBBMC} would {\chengB produce} more branches than the \texttt{VBBMC} algorithms from the initial branch (i.e., $m$ branches v.s. $n$ branches), the sizes of the branches are typically much smaller, which rarely reaches the worst case. Second, the early-termination technique \texttt{ET} also {\cheng helps} quickly enumerate those maximal cliques in a dense subgraph instead of making branches for the search space, which dramatically reduces the running time.
\texttt{RRcd} and \texttt{RDegen} also perform well in practice. Both are BK branching-based algorithms, where the former follows \texttt{BK\_Rcd} framework while the latter follows \texttt{BK\_Degen} framework. The main difference between \texttt{BK\_Rcd} and \texttt{BK\_Degen} is that \texttt{BK\_Rcd} adopts a \emph{top-to-down} approach that starts from a subgraph and iteratively removes vertices from the subgraph until the subgraph becomes a clique, while \texttt{BK\_Degen}, in essence, adopts a \emph{bottom-to-up} approach that starts from an individual vertex, adds only one vertex into the result set with one independent recursive call, and discovers a maximal clique utill all its vertices are added into the result set. When $g_C$ is already a clique or a near-clique (e.g., a clique with only a few edges missing), \texttt{BK\_Rcd} would create less branches than that of \texttt{BK\_Degen}. This would help explain why \texttt{BK\_Rcd} sometimes performs better than \texttt{BK\_Degen} in practice.


\smallskip\noindent
\textbf{(2) Ablation studies.}
We compare two variants of our methods, namely \texttt{HBBMC++} (the full version) and \texttt{HBBMC+} (the full version without the early-termination \texttt{ET}), with the state-of-the-art \texttt{VBBMC} algorithm \texttt{RDegen}. 
We note that it is not fair to compare \texttt{HBBMC} with other \texttt{VBBMC} algorithm (e.g., \texttt{RRcd}) since the former applies the classic pivot selection introduced in \cite{tomita2006worst} while the latter does not. 
We report the results in the columns 1-3 of Table~\ref{tab:ablation} and have the following observations. 
\underline{First}, \texttt{HBBMC+} runs consistently and clearly faster than \texttt{RDegen}, which demonstrates the contribution of the hybrid frameworks of \texttt{HBBMC} since \texttt{HBBMC+} and \texttt{RDegen} only differ in their frameworks (note that the graph reduction technique \texttt{GR} is employed in \texttt{RDegen}). 
\underline{Second}, \texttt{HBBMC++} runs faster than \texttt{HBBMC+}, which demonstrates the contribution of the early-termination technique \texttt{ET}. 
For example, on \textsf{OR}, the hybrid framework and early-termination technique contribute 41.7\% and 58.3\% to the efficiency improvements over \texttt{RDegen}, respectively.


\begin{table*}[t]
\vspace{-5mm}
    \centering
    \caption{Effects of early-termination technique, where ``Ratio'' is the ratio between the number of branches that can be early-terminated (i.e., $g_C$ is $t$-plex and $g_X$ is an empty graph) and the number of all branches where the condition that $g_C$ is a $t$-plex is satisfied.}
    \vspace{-2mm}
    \label{tab:early}
\begin{tabular}{c|R{1.1cm}R{1.1cm}|R{1.1cm}R{1.1cm}R{1.1cm}|R{1.1cm}R{1.1cm}R{1.1cm}|R{1.1cm}R{1.1cm}R{1.1cm}}
\hline
\multirow{2}{*}{} & \multicolumn{2}{c|}{$t=0$} & \multicolumn{3}{c|}{$t=1$} & \multicolumn{3}{c|}{$t=2$} & \multicolumn{3}{c}{$t=3$} \\\cline{2-12}
 & Time (s) & \#Calls  & Time (s) & \#Calls & Ratio & Time (s) & \#Calls & Ratio & Time (s) & \#Calls & Ratio \\ \hline
\textsf{NA} & 0.42 & 552K & 0.38 & 374K & 19.47\% & {0.34} & 366K & 19.83\% & {0.33} & 365K & 19.72\% \\
\textsf{FB} & 1.40 & 4.08M & 1.33 & 3.45M & 75.47\% & {1.25} & 2.77M & 74.90\% & {1.10} & 2.15M & 65.92\% \\
\textsf{WE} & 0.06 & 321K & 0.04 & 217K & 59.14\% & {0.03} & 206K & 59.23\% & {0.02} & 205K & 57.39\% \\
\textsf{WK} & 0.78 & 3.36M & 0.64 & 2.76M & 84.72\% & {0.60} & 2.20M & 83.06\% & {0.57} & 1.76M & 76.35\% \\
\textsf{SH} & 0.88 & 2.31M & 0.66 & 1.72M & 53.98\% & {0.52} & 1.66M & 52.91\% & {0.45} & 1.57M & 49.47\% \\
\textsf{ST} & 1.45 & 2.61M & 1.36 & 1.99M & 64.94\% & {1.29} & 1.76M & 60.98\% & {1.26} & 1.69M & 57.17\% \\
\textsf{DB} & 0.38 & 993K & 0.29 & 571K & 57.71\% & {0.22} & 550K & 55.92\% & {0.16} & 537K & 52.47\% \\
\textsf{DE} & 5.53 & 2.26M & 4.94 & 1.30M & 4.23\% & {4.02} & 1.29M & 4.63\% & {3.82} & 1.29M & 4.63\% \\
\textsf{DG} & 521.98 & 2.36B & 419.62 & 2.06B & 73.76\% & {347.80} & 1.78B & 71.38\% & {239.58} & 1.54B & 64.50\% \\
\textsf{YO} & 1.92 & 6.30M & 1.74 & 5.00M & 82.16\% & {1.57} & 4.37M & 78.91\% & {1.47} & 3.97M  & 74.79\% \\
\textsf{PO} & 22.33 & 38.6M & 21.20 & 33.4M & 63.58\% & {20.03} & 30.2M & 61.97\% & {19.31} & 27.9M & 57.25\% \\
\textsf{SK} & 40.81 & 102M & 35.83 & 82.3M & 82.65\% & {30.45} & 69.1M & 83.11\% & {25.15} & 53.8M & 77.81\% \\
\textsf{CN} & 7.50 & 25.1M & 6.86 & 20.6M & 78.74\% & {6.57} & 18.2M & 76.07\% & {6.03} & 16.6M & 71.92\% \\
\textsf{BA} & 18.73 & 36.4M & 15.49 & 31.1M & 73.33\% & {14.39} & 27.6M & 71.19\% & {13.81} & 25.1M & 66.88\% \\
\textsf{OR} & 1433.02 & 8.99B & 1034.83 & 7.73B & 69.29\% & {966.22} & 6.78B & 67.12\% & {884.20} & 5.58B & 62.90\%\\
\textsf{SO} & 32.16 & 63.9M & 26.03 & 53.2M & 71.11\% & {18.22} & 47.0M & 68.45\% & {21.12} & 42.5M & 62.69\% \\ \hline
\end{tabular}
\vspace{-4mm}
\end{table*}

\begin{table}[t]
\vspace{-2mm}
\scriptsize
    \centering
    \caption{Effects of truss-based edge ordering \\(unit: second).}
\vspace{-2mm}
    \label{tab:revise1}
    \begin{tabular}{c|rrrrr}
    \hline
    & \texttt{HBBMC++} & \texttt{VBBMC-dgn} & \texttt{HBBMC-dgn} & \texttt{HBBMC-mdg} \\ \hline
\textsf{NA} & \textbf{0.33}     & 0.44 & 0.45 & \underline{0.37}   \\
\textsf{FB} & \textbf{1.10}     & 1.42 & 1.43 & \underline{1.26}   \\
\textsf{WE} & \textbf{0.02}     & \underline{0.04} & 0.04 & 0.05   \\
\textsf{WK} & \textbf{0.57}     & 0.76 & 0.77 & \underline{0.73}   \\
\textsf{SH} & \textbf{0.45}     & 0.66 & 0.68 & \underline{0.55}   \\
\textsf{ST} & \textbf{1.26}     & 1.81 & 1.89 & \underline{1.57}   \\
\textsf{DB} & \textbf{0.16}     & 0.27 & 0.28 & \underline{0.23}   \\
\textsf{DE} & \textbf{3.82}     & 6.81 & 6.96 & \underline{5.13}   \\ 
\textsf{DG} & \textbf{239.58}   & 594.27 & 596.55 & \underline{486.02}   \\
\textsf{YO} & \textbf{1.47}     & \underline{2.42} & 2.51 & 2.53   \\
\textsf{PO} & \textbf{19.31}    & 25.99 & 26.58 & \underline{20.64}   \\
\textsf{SK} & \textbf{25.15}    & 37.58 & 38.71 & \underline{32.30}   \\
\textsf{CN} & \textbf{6.03}     & 11.91 & 12.36 & \underline{7.83}  \\
\textsf{BA} & \textbf{13.81}    & 16.78 & 17.19 & \underline{16.58}  \\
\textsf{OR} & \textbf{884.20}   & 1505.95 & 1550.6 & \underline{1204.22}   \\
\textsf{SO} & \textbf{21.12}    & 36.03 & 37.33 & \underline{27.66}  \\
\hline
    \end{tabular}
\vspace{-4mm}
\end{table}

\smallskip\noindent
\textbf{(3) Effects of hybrid frameworks (varying vertex-oriented branching framework).}
Recall that in Section~\ref{subsec:HBBMC}, there are several combinations of \texttt{EBBMC} and \texttt{VBBMC} to implement the hybrid framework (i.e., using different \texttt{VBBMC} frameworks). Therefore, we first study the effect of these combinations by making comparison among these implementations. We denote the algorithms for comparison by \texttt{HBBMC++} (the full version), \texttt{Ref++}, \texttt{Rcd++} and \texttt{Fac++}. We note that each of the last three algorithms corresponds to an enumeration procedure in which {\cheng we} (1) use the edge-oriented BB framework with truss-based edge ordering {\cheng for the partitioning at} the initial branch, (2) use the corresponding \texttt{VBBMC} algorithm (e.g., \texttt{BK\_Ref}, \texttt{BK\_Rcd} or \texttt{BK\_Fac}) to produce the remaining branches and (3) use \texttt{ET} (with $t=3$) and \texttt{GR} to boost the performance. We report the results in columns 4-6 of Table~\ref{tab:ablation}. Compared among the algorithms with different \texttt{VBBMC} frameworks,  \texttt{HBBMC++} runs the fastest on the majority of the datasets. While \texttt{Rcd++} runs faster on some datasets (e.g., \textsf{FB}, \textsf{SH} and \textsf{DE}), the differences are marginal.

\smallskip\noindent
\textbf{(4) Effects of hybrid frameworks (varying the time when \texttt{HBBMC++} changes from edge-oriented to vertex-oriented branching strategy).}
We also explore alternative hybrid strategies by investigating the use of the edge-oriented branching strategy at varying depths within the recursion tree. Specifically, we applied the edge-oriented branching strategy at the early stages of the tree, up to a depth threshold $d\ge 1$, beyond which we switched to the pivot-based vertex-oriented branching strategy. By varying $d$ within the range $d\in\{1,2,3\}$, we aimed to identify the optimal balance between the effectiveness of edge-oriented branching and the more granular control provided by vertex-oriented branching at deeper levels. We note that the case $d=1$ corresponds to the branching strategy employed in \texttt{HBBMC++}.
Table \ref{tab:hybrid} shows that increasing $d$ typically results in more branching calls and longer runtimes, due to the lack of pivot-based pruning in edge-oriented branching. However, edge-oriented branching is still effective for rapid pruning at the initial stages, and adjusting $d$ allows us to fine-tune when to leverage the strengths of each strategy. The results suggest that carefully choosing the threshold $d$ (i.e., $d=1$) is crucial for optimizing the performance of the hybrid framework, ensuring that the transition between branching strategies is both timely and effective.

\smallskip\noindent
\textbf{(5) Effects of early-termination technique.}
We study the effects by choosing different parameter $t$ (in $t$-plex) for early-termination. We vary $t$ in the range $\{0, 1, 2, 3\}$ and report the corresponding running time of \texttt{HBBMC++} under different values of $t$. We note that when $t=0$, \texttt{HBBMC++} is equivalent to \texttt{HBBMC+} (the full version without early-termination technique). The results are shown in Table~\ref{tab:early}. We observe that the running time of \texttt{HBBMC++} decreases as the value of $t$ increases. The reason is that 
we can early-terminate the recursion in an earlier phase when $t$ has a larger value as $t$-plex must be a $t'$-plex if $t<t'$ but not vice versa. We collect the statistics of the number of recursive function calls of \textsf{VBBMC\_Rec} in \texttt{HBBMC}, showing that the numbers drop steadily from $t=0$ to $t=3$, which explains the contribution of the efficiency improvement with different values of $t$. 
Besides, there is another factor that affects the power of the early-termination. That is the ratio between the number of branches that can be early-terminated (i.e., $g_C$ is $t$-plex and $g_X$ is an empty graph), which we denote by $b_0$, {\cheng and} the number of all branches where the condition that $g_C$ is a $t$-plex is satisfied, which we denote by $b$. The closer the ratio is to 1, the more potent the early-termination becomes. We also report this ratio in Table~\ref{tab:early}. We have the following observations. First, the ratio is over 60\% with different settings of $t$ on most of datasets, which indicates that the early-termination truly helps \texttt{HBBMC} accelerate the enumeration in such dense subgraphs. For the datasets such as \textsf{NA} and \textsf{DE} in which the ratio is relatively small, the possible reason is that these two graphs are dense (i.e., $\rho > 20$) but have too few maximal cliques inside (i.e., there are only 28,154 and 24,012 maximal cliques in \textsf{NA} and \textsf{DE}, respectively). As a result, the value of $b$ is much larger than $b_0$ on these two graphs. Second, it shows a trend that the ratio decreases as the value of $t$ increases. We collect the values of $b_0$ and $b$ and find that while $b_0$ and $b$ both increase as $t$ increases, the increase rate of $b$ is larger than that of $b_0$. The possible reason is that the requirement that $g_C$ is $t$-plex is much easier to be satisfied as the value of $t$ increases but $g_X$ is not always an empty graph.


\smallskip\noindent
\textbf{(6) Effects of truss-based edge ordering.}
We compare the performance of \texttt{HBBMC++} with three new baselines that only differ in their branching strategies for the initial branch. \texttt{VBBMC-dgn} uses a vertex-oriented branching strategy based on degeneracy ordering. \texttt{HBBMC-dgn} adopts an edge-oriented branching strategy where edges are ordered alphabetically based on the degeneracy ordering of their endpoints. \texttt{HBBMC-mdg} also uses an edge-oriented branching strategy, but the edge ordering is based on a non-decreasing order of an upper bound derived from the minimum degree of their endpoints.
The results are shown in Table~\ref{tab:revise1}. We have the following observations. \underline{First}, \texttt{HBBMC++} (with the proposed hybrid branching) outperforms \texttt{VBBMC-dgn} (with the existing vertex-oriented branching) across all datasets, with significant improvements particularly on larger datasets. This demonstrates the superiority of the porposed hybrid branching, particularly the edge-oriented branching, and is well aligned with the theoritical results in Theorem~\ref{theo:hbbmc}. \underline{Second}, \texttt{HBBMC++} consistently outperforms both \texttt{HBBMC-dgn} and \texttt{HBBMC-mdg}, which verifies the effectiveness of truss-based edge ordering towards reducing computational overhead. We remark that, without adopting the truss-based edge ordering, \texttt{HBBMC-dgn} and \texttt{HBBMC-mdg} cannot acheive the same time complexity as \texttt{HBBMC++} does.

\section{Related Work}
\label{sec:related}

%
Sequential algorithms for maximal clique enumeration (MCE) problem 
can be classified into two categories: branch-and-bound based algorithms \cite{bron1973algorithm, tomita2006worst, eppstein2010listing, eppstein2013listing, li2019fast, naude2016refined} and reverse-search based algorithms \cite{johnson1988generating, chang2013fast, conte2016sublinear, tsukiyama1977new, chiba1985arboricity, makino2004new, comin2018improved}. 
%
{\chengc For the former,}
Bron-Kerbosch (\texttt{BK}) algorithm \cite{bron1973algorithm} is the first practical algorithm for MCE, in which each branch is represented by three vertex sets $S$, $C$ and $X$ ({\cheng correspondingly,} $g_C=G[C]$ and $g_X=G[X]$ are the candidate and {\kaixin exclusion}
subgraphs,
respectively). \texttt{BK\_Pivot} \cite{tomita2006worst} improves \texttt{BK} {\cheng with} the pivoting strategy, which selects a vertex $u\in C\cup X$ that maximizes $|N(u,G)\cap C|$ as the pivot in each branch. It has been {\cheng proved} that the worst-case time complexity of \texttt{BK\_Pivot} is $O(n\cdot 3^{n/3})$, which is worst-case optimal since the Moon-Moser graph \cite{moon1965cliques} truly has $3^{n/3}$ maximal cliques inside. \texttt{BK\_Ref} \cite{naude2016refined} refines and accelerates \texttt{BK\_Pivot} by considering some special cases during the pivot selection {\cheng process} but retains the same worst-case time complexity as that of \texttt{BK\_Pivot}. Recently, different vertex orderings are adopted to improve \texttt{BK\_Pivot} on sparse graphs. Specifically, \texttt{BK\_Degen} \cite{eppstein2010listing, eppstein2013listing} uses the degeneracy ordering {\cheng for the branching at the initial branch} $B=(\emptyset, G, G[\emptyset])$ such that for each produced sub-branch, there are at most $\delta$ vertices in the candidate subgraph $g_C$, where $\delta$ is the degeneracy of the graph, and correspondingly, the worst-case time complexity improves to $O(n\delta\cdot 3^{\delta/3})$. Similarly, \texttt{BK\_Degree} \cite{xu2014distributed} uses the degree ordering to branch the initial search space. In this case, there are at most $h$ vertices in the candidate subgraph $g_C$, where $h$ is the $h$-index of graph, and the worst-case time complexity is $O(nh\cdot 3^{h/3})$. \texttt{BK\_Rcd} \cite{li2019fast} extends \texttt{BK\_Degen} {\cheng with} a top-to-down approach, {\chengB which} repeatedly conducts branching at the vertex with the smallest degree in $g_C$ until $g_C$ {\cheng becomes} a clique. \texttt{BK\_Fac} \cite{jin2022fast} aims to choose the pivot efficiently. Given a branch $B$, instead of finding an optimal pivot vertex that could minimize the number of produced sub-branches before the branching steps, \texttt{BK\_Fac} first {\cheng initializes} an arbitrary vertex in $g_C$ as pivot and iteratively changes to a new one if branching on it would produce fewer branches. However, given a branch $B=(S,g_C,g_X)$, \texttt{BK\_Rcd} and \texttt{BK\_Fac} would produce $2^{|V(g_C)|}$ sub-branches in total in the worst-case, the time complexities of these two algorithms are $O(n\delta \cdot 2^{\delta})$. In \cite{deng2023accelerating}, several graph reduction rules are proposed to improve the performance pratically. We note that these rules are orthogonal to frameworks. 
%

{\chengc For the latter,}
they can be viewed as different instances of the general reverse-search enumeration framework \cite{avis1996reverse}. {\cheng It} traverses from one solution to another solution following a depth-first search (DFS) procedure on an implicit connected solution graph. 
%
Although reverse-search based algorithms can guarantee that the time cost between two consecutive outputs is bounded by a polynomial function of the input size (i.e., polynomial delay), which implies that the time complexity of these algorithm is proportional to the number of maximal cliques within the input graph (i.e., output-sensitive), their performances on real-world graphs are not as {\cheng competitive} as those of branch-and-bound based algorithms \cite{das2018shared}. 

There are also numerous studies that develop parallel algorithms for MCE \cite{han2018speeding, inoue2014faster, wu2009distributed, schmidt2009scalable, lu2010dmaximalcliques, xu2015distributed, abu2017genome, lessley2017maximal, du2006parallel, cheng2012fast, blanuvsa2020manycore, das2020shared, brighen2019listing, chen2016parallelizing, hou2016efficient, svendsen2015mining}.
These algorithms extend the techniques that are originally proposed for sequential algorithms and aim to balance and reduce the computation overheads to achieve higher degrees of parallelism. Besides, there are also several related studies that (1) solve {\cheng the} MCE problem on other graph types, such as uncertain graphs \cite{dai2022fast, mukherjee2015mining}, temporal graphs \cite{himmel2016enumerating}, attributed graphs \cite{pan2022fairness, zhang2023fairness} and graph streams \cite{das2019incremental, sun2017mining}, and (2) solve MCE's counterpart problem, such as maximal biclique enumeration on bipartite graph \cite{chen2022efficient, abidi2020pivot} and other maximal subgraph enumeration problems \cite{yu2023fast}. We note that the techniques proposed for these problems are not suitable for general MCE problem.

\section{Conclusion}
\label{sec:conclusion}

In this paper, we study the maximal clique enumeration (MCE) problem. 
We propose a hybrid branch-and-bound (BB) algorithm, \texttt{HBBMC}, integrating both vertex-oriented and edge-oriented branching strategies, achieving a time complexity 
{\chengc which is better than the best-known one under a condition that holds for the majority of real-world grpahs.}
To further boost efficiency, we introduce an early termination technique, {\chengB which is} applicable across BB frameworks. Extensive experiments demonstrate the efficiency of our algorithms and the effectiveness
of proposed techniques.

\smallskip
\noindent\textbf{Acknowledgments.}
This research is supported by the Ministry of Education, Singapore, under its Academic Research Fund (Tier 2 Award MOE-T2EP20221-0013, Tier 2 Award MOE-T2EP20220-0011, and Tier 1 Award (RG20/24)). Any opinions, findings and conclusions or recommendations expressed in this material are those of the author(s) and do not reflect the views of the Ministry of Education, Singapore.

\appendices


\section{Time Complexities of \texttt{VBBMC} Algorithms}

\label{sec:app-vbb-time}

\begin{table}[h]
\small
    \centering
    \caption{The worst-case time complexities of existing \texttt{VBBMC} algorithms, where $n$, $h$ and $\delta$ are the number of the vertices, the $h$-index and the degeneracy of the graph, respectively.}
    \begin{tabular}{c|c}
    \hline
       Algorithm name & Time complexity \\ \hline
        \texttt{BK} \cite{bron1973algorithm} &  $O(n \cdot (3.14)^{n/3})$ \\
        \texttt{BK\_Pivot} \cite{tomita2006worst}, \texttt{BK\_Ref} \cite{naude2016refined} & $O(n \cdot 3^{n/3})$ \\
        \texttt{BK\_Degree} \cite{xu2014distributed} & $O(hn\cdot 3^{h/3})$ \\
        \texttt{BK\_Degen} \cite{eppstein2010listing, eppstein2013listing} & $O(\delta n \cdot 3^{\delta / 3})$ \\
        \texttt{BK\_Rcd} \cite{li2019fast}  & $O(\delta n \cdot 2^{\delta})$ \\ 
        \texttt{BK\_Fac} \cite{jin2022fast} & $O(\delta n \cdot (3.14)^{\delta/3})$\\ \hline
    \end{tabular}
    \label{tab:complexity}
\end{table}

Different variants of \texttt{VBBMC} have different time complexities, which we summarize in Table~\ref{tab:complexity}. We note that \texttt{BK\_Rcd} \cite{li2019fast} and \texttt{BK\_Fac} \cite{jin2022fast} do not provide the worst-case time {\cheng complexity} analysis in their paper. We provide 
{\cheng the analysis} as follows.

\subsection{Time Complexity of \texttt{BK\_Rcd}}

We present the algorithm details of \texttt{BK\_Rcd} \cite{li2019fast} in Algorithm~\ref{alg:bkrcd}. Each branch is represented by three vertex sets $S$, $C$ and $X$ (correspondingly, $g_C=G[C]$ and $g_X=G[X]$ are the candidate and exclusion subgraph as introduced in Section~\Romannum{3}). 

\begin{theorem}
    Given a graph $G=(V,E)$, the worst-case time complexity of \texttt{BK\_Rcd} is $O(n \delta \cdot 2^{\delta})$, where $\delta$ is the degeneracy of the graph.
\end{theorem}


\begin{proof}
    From line 6 and lines 8-9, we observe that, given a branch $B=(S,C,X)$, the algorithm iteratively removes the vertex with the minimum degree in the remaining candidate subgraph $G[C]$. This corresponds to using degeneracy ordering of $G[C]$ for branching. Let $T(c,x)$ be the total time cost to enumerate all maximal cliques within a branch $B=(S,C,X)$, where $|C|=c$ and $|X|=x$. Consider the initial branch. Based on the degeneracy ordering, the candidate subgraph of each produced sub-branch have at most $\delta$ vertices. Consider the branches other than the initial branch. While removing the vertex with the minimum degree is intuitive enough such that the remaining graph $G[C]$ is possibe to be a non-trivial clique (other than a vertex or an edge) at line 10, it is not always the case. Consider an example as illustrated in Figure~\ref{fig:worst}. The number in each vertex represents the index of the vertex. It is easy to check that the sequence $\{v_1, v_2, v_3, v_4, v_5, v_6, v_7, v_8\}$ is a valid degeneracy ordering. However, only when $C$ has $\{v_7,v_8\}$ inside, it is a clique (i.e., an edge). 
    All above examples demonstrate that given a branch $B$, the \texttt{BK\_Rcd} would produce $|C|-2$ branches at the worst case. Thus, we have the following recurrence. 
    \begin{equation}
T(c, x) \le \left\{
\begin{array}{ll}
    O(\delta) & c = 0\\ 
    O(|E|) + \sum_{v\in V} T(\delta, \Delta) &  c = |V| \\ 
    \sum_{i=2}^{c-1} T(i, x) & c\neq 0 \\
\end{array}
\right.
\end{equation}
By solving the recurrence, we have the time complexity of \texttt{BK\_Rcd}, which is $O(\delta n \cdot 2^{\delta})$. 
\end{proof}

\begin{algorithm}[th]
\small
\caption{\texttt{BK\_Rcd}~\cite{li2019fast}}
\label{alg:bkrcd}
\KwIn{A graph $G=(V,E)$}
\KwOut{All maximal cliques within $G$}
\SetKwFunction{List}{\textsf{BK\_Rcd\_Rec}}
\SetKwProg{Fn}{Procedure}{}{}
\List{$\emptyset, V, \emptyset$}\;
\Fn{\List{$S, C, X$}} {
    \tcc{Termination if $C, X$ are empty}
    \If{$C\cup X=\emptyset$}{
        \textbf{Output} a maximal clique $S$; \Return\;
    }
    \tcc{Branching when $G[C]$ is not a clique}
    \While{$G[C]$ is not a clique}{
        Choose $v\in C$ that minimizes $|N(u, G[C])|$\;
        \List{$S\cup \{v\}, C\cap N(v, G), X\cap N(v, G)$}\;
        $C\gets C\setminus \{v\}$\;
        $X\gets X\cup \{v\}$\;
    }
    \tcc{Check maximality}
    \If{$C\neq \emptyset$ and $N(C, G)\cap X=\emptyset$}{
            \textbf{Output} a maximal clique $S\cup C$\;
    }
    
}
\end{algorithm}

\begin{figure}[th]
    \centering
    \includegraphics[width=0.25\textwidth]{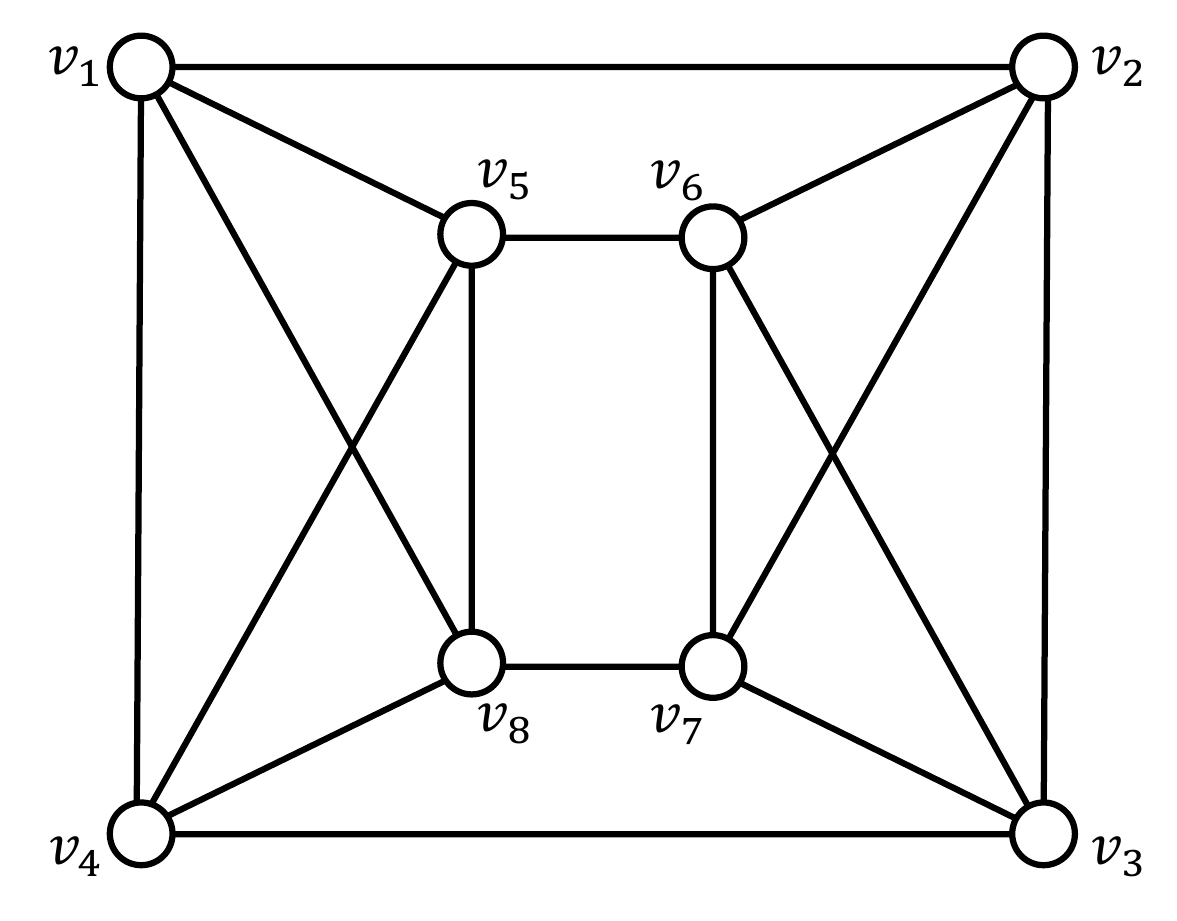}
    \caption{A worst case for \textsf{BK\_Rcd} algorithm.}
    \label{fig:worst}
\end{figure}

\begin{algorithm}[t]
\caption{\texttt{BK\_Fac}~\cite{jin2022fast}}
\label{alg:bkfac}
\KwIn{A graph $G=(V,E)$}
\KwOut{All maximal cliques within $G$}
\SetKwFunction{List}{\textsf{BK\_Fac\_Rec}}
\SetKwProg{Fn}{Procedure}{}{}
Let $v_1, \cdots, v_n$ be in the degeneracy ordering of $G$\;
\For{each $v_i\in V$}{
$C_i\gets N(v_i, G)\cap \{v_1, \cdots, v_{i-1}\}$\;
$X_i\gets N(v_i, G)\cap \{v_{i+1}, \cdots, v_n\}$\;
\List{$\{v_i\}, C_i, X_i$}\;
}
\Fn{\List{$S, C, X$}} {
    \tcc{Termination if $C, X$ are empty}
    \If{$C\cup X=\emptyset$}{
        \textbf{Output} a maximal clique $S$; \Return\;
    }
    \tcc{Initialize a pivot $v$, create branches on $P$}
    $v \gets$ an arbitrary vertex in $C$\;
    $P \gets C \setminus N(v, G)$\; 

    \For{each $u\in P$}{
        \List{$S\cup \{u\}, C\cap N(u, G), X\cap N(u,G)$}\;
        $C\gets C\setminus \{u\}$\;
        $X\gets X\cup \{u\}$\;
        \tcc{Update $P$ if possible}
        $P\gets P\setminus \{u\}$\;
        $P'\gets C\setminus N(u, G)$\;
        \lIf{$|P'|<|P|$}{$P\gets P'$}

    }
}
\end{algorithm}

\subsection{Time Complexity of \texttt{BK\_Fac}}

We present the algorithm details of \texttt{BK\_Fac} \cite{jin2022fast} in Algorithm~\ref{alg:bkfac}. Similarly, each branch is represented by three sets $S$, $C$ and $X$ (correspondingly, $g_C=G[C]$ and $g_X=G[X]$ are the candidate and exclusion subgraph as introduced in Section~\Romannum{3}).

\begin{theorem}
    Given a graph $G=(V,E)$, the worst-case time complexity of \texttt{BK\_Fac} is $O(\delta n \cdot (3.14)^{\delta/3})$, where $\delta$ is the degeneracy of the graph. 
\end{theorem}

\begin{proof}
    Based on the degeneracy ordering, each produced sub-branch from line 5 will have at most $\delta$ vertices in $C_i$, where $\delta$ is the degeneracy of the graph. Then consider the time complexity of the \textsf{BK\_Fac\_Rec}. 
    We have the following observations. First, the initial pivot is chosen from $C$ (line 9) and $P$ is the vertex set that includes those vertices at which the branching steps are conducted (line 10). Second, although lines 15-17 would update $P$ to reduce the number of produced branches, there exists the worst case in which $P$ is not changed. This worst case is the same as the second algorithm proposed in \cite{bron1973algorithm}, whose time complexity is $O(|C|\cdot(3.14)^{|C|/3})$ given a branch $B=(S,C,X)$. In summary, the time complexity of \textsf{BK\_Fac} is $O(\delta n \cdot (3.14)^{\delta/3})$. 
\end{proof}

\section{Proof of Theorem 2}
\label{app:proof-theorem2}

\begin{proof}
The running time of \texttt{HBBMC} consists of (1) the time of generating the truss-based edge ordering and obtaining the subgraphs (i.e., $g_C$ and $g_X$) for enumeration, which can be done in $O(\delta m)$ time, where $\delta$ is the degeneracy of the graph, and (2) the time of recursive enumeration procedure.
    
\underline{Consider the former}. First, the time complexity of generating the truss-based edge ordering is $O(\delta m)$ \cite{wang2024technical}. Besides, the time of getting the subgraphs can also be bounded by $O(\delta m)$. The reason is as follows. First, for all edges $e=(u,v)\in E$, we can compute the set of common neighbors of $u$ and $v$ (i.e., $N(V(e), G)$ for $e\in E$) in $O(\delta m)$ time by any existing triangle listing algorithm \cite{hu2014efficient,chu2011triangle,latapy2008main}. Second, consider an edge $e=(u,v)\in E$. For each $w\in N(V(e), G)$, $w\in V(g_C)$ if both edges $(u, w)$ and $(v, w)$ are ordered after $e$; otherwise, $w\in V(g_X)$. This process can also be done in $O(\delta m)$. To see this, we note that $u$, $v$ and $w$ always form a triangle (i.e., 3-clique), indicating that the time cost for getting all subgraphs for the sub-branches is proportional to the number of triangles in $G$. According to the existing triangle listing algorithms \cite{hu2014efficient,chu2011triangle,latapy2008main}, given a graph with its degeneracy number equal to $\delta$, the number of triangles inside is bounded by $O(\delta m)$. With $V(g_C)$ and $V(g_X)$, it is sufficient to recursively conduct the vertex-oriented branching strategy for the remaining branches.

\underline{Consider the latter}, we analyze the time costs for the recursive process. We first show two lemmas that is necessary for the proof. 

\begin{lemma}[Theorem 3 of \cite{tomita2006worst}]
\label{lemma:1}
Let $F(n)$ be a function which satisfies the following recurrence relation
\begin{equation}
    F(n)\le\left\{
\begin{array}{lr}
    \max_{k} \{k\cdot F(n-k)\} + O(n^2) & \quad \text{if } n > 0 \\
     O(1) & \quad \text{if } n=0
\end{array}
    \right.
\end{equation}
where $n$ and $k$ are positive integers. Then $F(n) = O(3^{n/3})$.
\end{lemma}

\begin{lemma}
\label{lemma:2}
    Given a graph $G=(V, E)$ with its truss number equal to $\tau$. The number of triangles in $G$ is no more than $m\tau$. 
\end{lemma}

\begin{proof}
    Let $\pi_{\tau}=\langle e_1, e_2, \cdots, e_m \rangle$ be the truss-based edge ordering and $sup(e_i)$ be the number of common neighbors in the subgraph induced by $\{e_i, e_{i+1}, \cdots, e_m\}$. Then, the number of triangles in $G$, denoted by $N$, is bounded by
    \begin{equation}
        N =\sum_{i=1}^m sup(e_i) \le \sum_{i=1}^m \tau =m\tau
    \end{equation}
    The first equation follows from the fact that each triangle is uniquely associated with the earliest edge in the truss-based ordering of its three constituent edges. The inequality follows from the definition of $\tau$. 
\end{proof}

Now, we are ready for the analysis. Consider a branch $B=(S,g_C,g_X)$ produced after the edge-oriented branching strategy, let $T(c,x)$ be the time complexity of solving $B$ with $c=|V(g_C)|$ and $x=|V(g_X)|$. 
The time cost consists of three components. First, it would take $O(c\cdot (c+x))$ to choose the pivot by checking each vertex in $V(g_C)\cup V(g_X)$ to identify the number of its neighbors in $g_C$. Second, assume that the pivot is $v^*$ and $v^*$ has $k$ ($k\le c$) non-neighbors in $g_C$ (denoted by $u_1, \cdots, u_k$), on which the sub-branches are produced. Correspondingly, for each sub-branch produced on $u_i$, it would take $O(c+x)$ to identify $u_i$'s neighbors in $g_C$ and $g_X$. In total, it would take $O(k\cdot (c+x))$ in total to construct $k$ sub-branches. Third, since $v^*$ is a pivot that has the maximum number of neighbors in $g_C$, i.e., $c-k$ neighbors, for each $u_i\in V(g_C)$, they all have at most $c - k$ neighbors in $g_C$. Thus, we have the following recurrence 
\begin{equation}
\label{eq:rec}
\begin{aligned}
    T(c,x) &\le \max_k \{k \cdot T(c-k,x)\} + \lambda \cdot c\cdot (c+x) \\
    &=(c+x) \cdot \Big( \max_k \{k \cdot \frac{T(c-k,x)}{c+x}\} + \lambda \cdot c \Big) \\
    &\le (c+x) \cdot \Big( \max_k \{k \cdot \frac{T(c-k,x)}{c-k+x}\} + \lambda \cdot c^2 \Big) \\
\end{aligned}
\end{equation}
The first inequality follows by summing up the three components of the time cost, where $\lambda>0$ is a constant. The equality follows by extracting $(c+x)$. To see the third inequality, we first have $\lambda\cdot c\le \lambda \cdot c^2$ since $c\ge 1$. Then, let $k^*$ be value of $k$ such that $k \cdot \frac{T(c-k,x)}{c+x}$ reaches the maximum. We have 
\begin{equation}
\begin{aligned}
    \max_k \{k \cdot \frac{T(c-k,x)}{c+x}\} &= k^* \cdot \frac{T(c-k^*,x)}{c+x} \\
    &\le k^* \cdot \frac{T(c-k^*,x)}{c-k^*+x} \\
    &\le  \max_k \{k \cdot \frac{T(c-k,x)}{c-k+x}\}
\end{aligned}
\end{equation}
Let $f(c)=\frac{T(c,x)}{(c+x)}$, we have 
\begin{equation}
    f(c)\le \max_k \{k \cdot f(c-k)\} + \lambda \cdot c^2
\end{equation}
According to Lemma~\ref{lemma:1}, we have $f(c)=\lambda'\cdot 3^{c/3}$, where $\lambda'>0$ is also a constant. Therefore, $T(c,x)=\lambda'\cdot(c+x)\cdot 3^{c/3}$. To compute the total time complexity, 
\begin{equation}
\begin{aligned}
    \sum_{e\in E} T(c_e, x_e) &=\sum_{e\in E} \lambda' \cdot (c_e + x_e) \cdot 3^{c_e/3}\\
    &\le \lambda' \cdot 3^{\tau/3} \cdot \sum_{e\in E} (c_e + x_e) \le \lambda'  \cdot 3^{\tau/3} \cdot 3m\tau
\end{aligned}
\end{equation}
The first inequality follows from the fact that $c_e\le \tau$ based on the definition of $\tau$. The second inequality follows from the facts that (1) for an edge $e$, $(c_e+x_e)$ is the number of triangles associated with $e$, and (2) the total number of triangles is no larger than $m\tau$ (Lemma~\ref{lemma:2}). In summary, the total time cost of \texttt{HBBMC} is $O(m \delta + m \tau \cdot 3^{\tau/3})$. 

\end{proof}

\section{Examples of Enumerating Maximal Cliques from 2-Plexes and 3-Plexes}
\label{app:example}

\begin{figure}[th]
\centering
\subfigure[A 2-plex graph $g_C$.]{
    \label{subfig:2-plex}
    \includegraphics[width=0.21\textwidth]{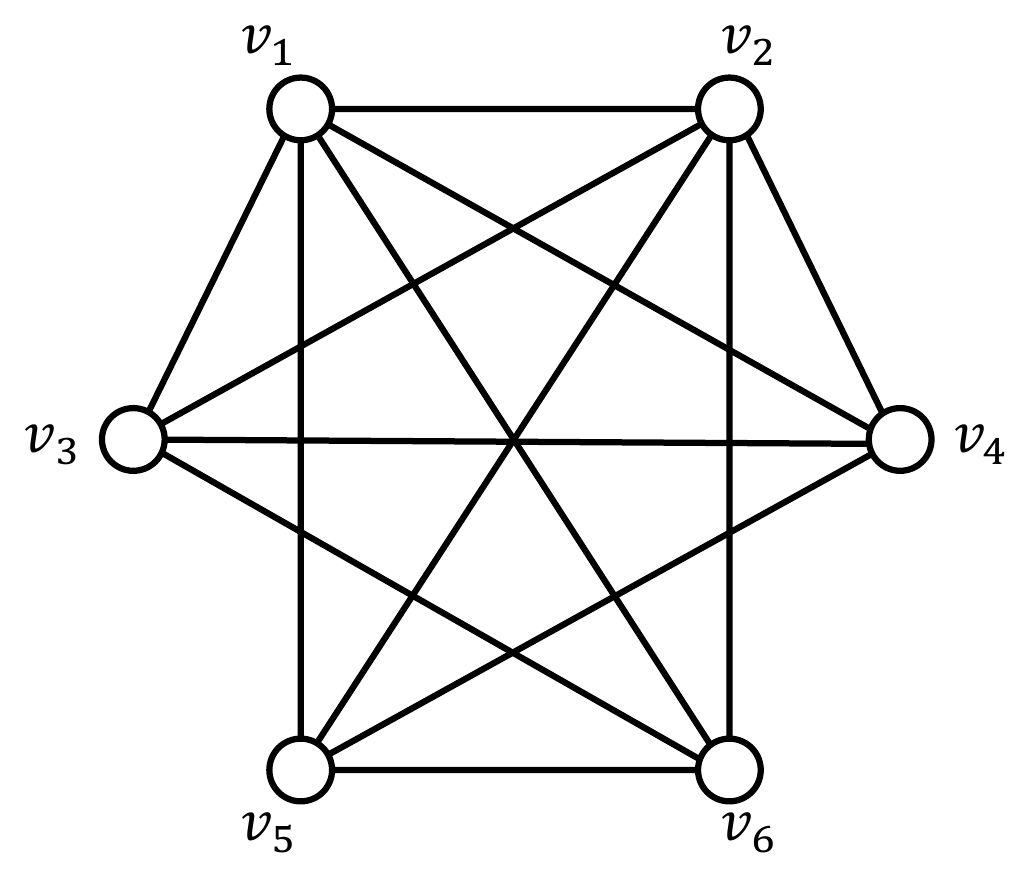}}
    \hspace{2mm}
\subfigure[Inverse graph $\overline{g}_C$ of $g_C$.]{
    \label{subfig:2-plex-inv}
    \includegraphics[width=0.21\textwidth]{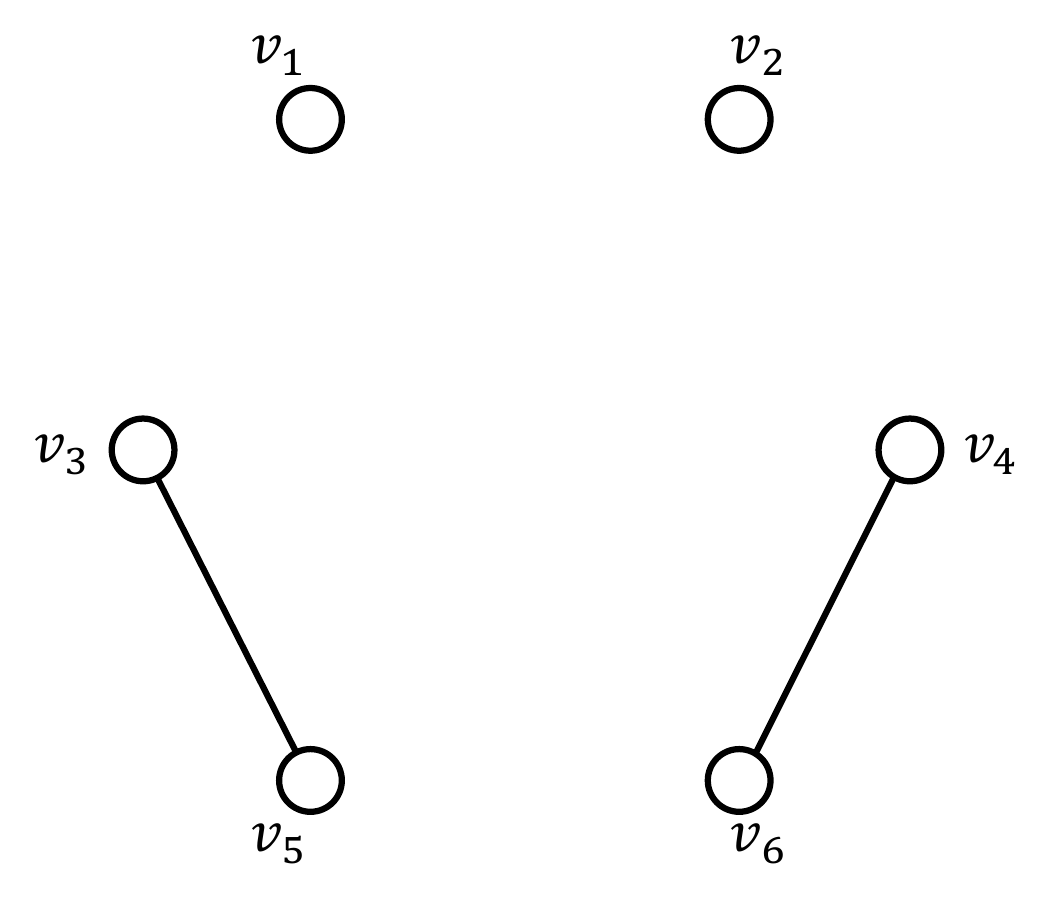}}
\caption{{Example of a 2-plex and its inverse graph.}} 
\label{fig:2-plex}
\end{figure}

\begin{figure}[th]
\centering
\subfigure[ 3-plex graph $g_C$.]{
    \label{subfig:3-plex}
    \includegraphics[width=0.21\textwidth]{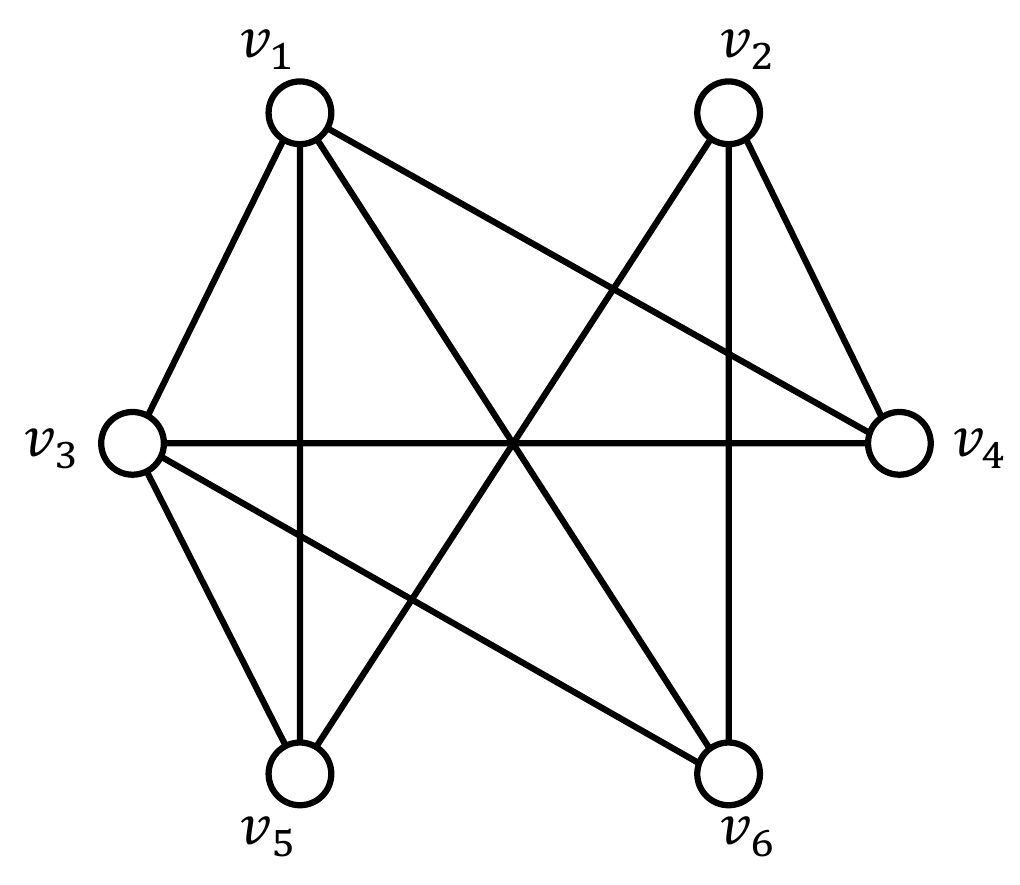}}
\hspace{2mm}
\subfigure[Inverse graph $\overline{g}_C$ of $g_C$.]{
    \label{subfig:3-plex-inv}
    \includegraphics[width=0.21\textwidth]{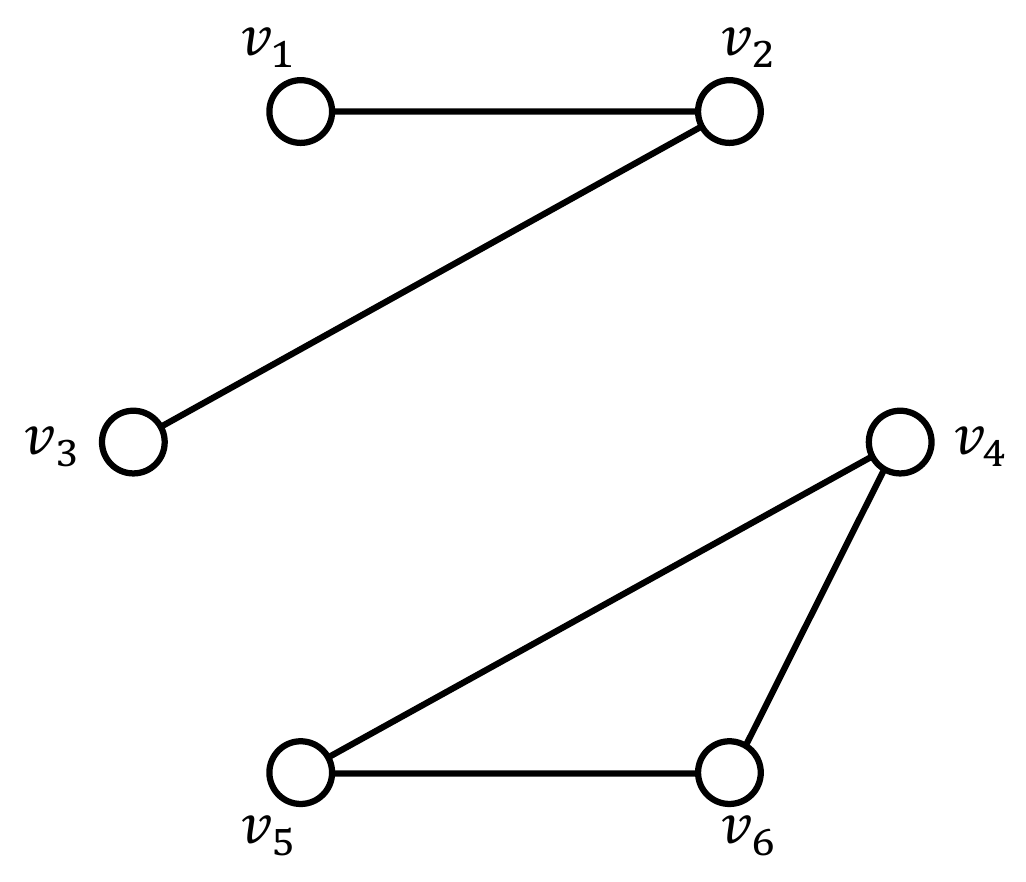}}
\caption{{Example of a 3-plex and its inverse graph.}} 
\label{fig:3-plex}
\vspace{-3mm}
\end{figure}

\smallskip
\noindent\textbf{Example.}
Consider the example of a 2-plex and its inverse graph in Figure~\ref{fig:2-plex}. Based on the definition of $F$, $L$ and $R$ as discussed above, one possible partition is $F=\{v_1, v_2\}$, $L=\{v_3,v_4\}$ and $R=\{v_5,v_6\}$. Given $|L|=2$,  we can directly output 4 maximal cliques as follows, $\{v_1, v_2, v_3, v_4\}$, $\{v_1,v_2,v_3, v_6\}$, $\{v_1, v_2, v_4, v_5\}$ and $\{v_1, v_2, v_5, v_6\}$. 

Consider the example of a 3-plex and its inverse graph in Figure~\ref{fig:3-plex}. Based on the definitions, the inverse graph contains a path $p=\{v_1, v_2, v_3\}$ and a cycle $c=\{v_4, v_5, v_6\}$. Consider the path $p$. It has two partial maximal cliques inside, i.e., $\{v_1,v_3\}$ and $\{v_2\}$. Consider the cycle $c$. It has three partial cliques inside, i.e., $\{v_4\}$, $\{v_5\}$ and $\{v_6\}$. Therefore, there are 6 maximal cliques in $g_C$, i.e., $\{v_1, v_3, v_4\}$, $\{v_1, v_3,v_5\}$, $\{v_1, v_3,v_5\}$, $\{v_2, v_4\}$, $\{v_2,v_5\}$ and $\{v_2, v_6\}$.

\section{Experimental Setup and Results on Synthetic Datasets}
\label{app:syn}

\smallskip\noindent\textbf{Setup.} 
The synthetic datasets are generated based on two random graph generators, namely the Erd\"{o}s-R\'{e}yni (ER) graph model \cite{erdds1959random} and the Barab\'{a}si–Albert (BA) graph model \cite{albert2002statistical}. 
Specifically, in the ER generator, the model first generates $n$ vertices and then randomly chooses $m$ edges between pairs of vertices. In the BA generator, the model adds the vertices to the network one at a time, and each new vertex is connected to a set number of existing vertices based on preferential attachment, i.e., vertices with a higher degree are more likely to receive new links.
By default, the number of vertices and edge density are set as  $n=1,000,000$ and $\rho=20$ (i.e., $m=n\rho$) for synthetic datasets, respectively. Due to the random nature of the graph generators, given a setting, we generate 10 graphs and report the mean of the results over these graphs.

\begin{figure}[th]
\centering
\subfigure[Scalability test (ER model).]{
\label{fig:scale-er}
\includegraphics[width=0.23\textwidth]{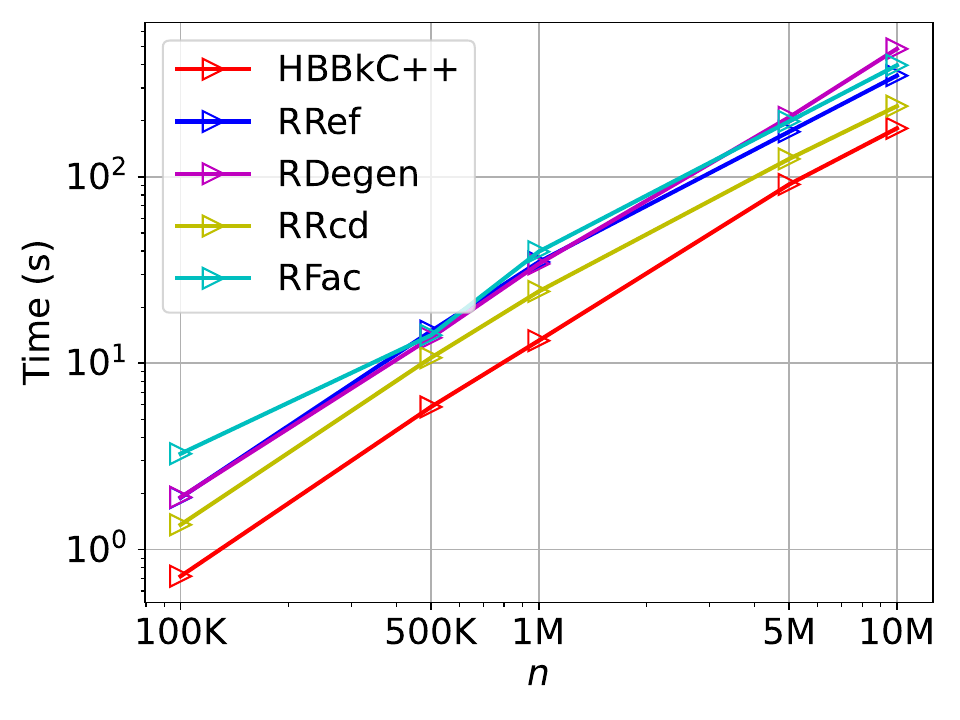}}
\subfigure[Scalability test (BA model).]{
\label{fig:scale-ba}
\includegraphics[width=0.23\textwidth]{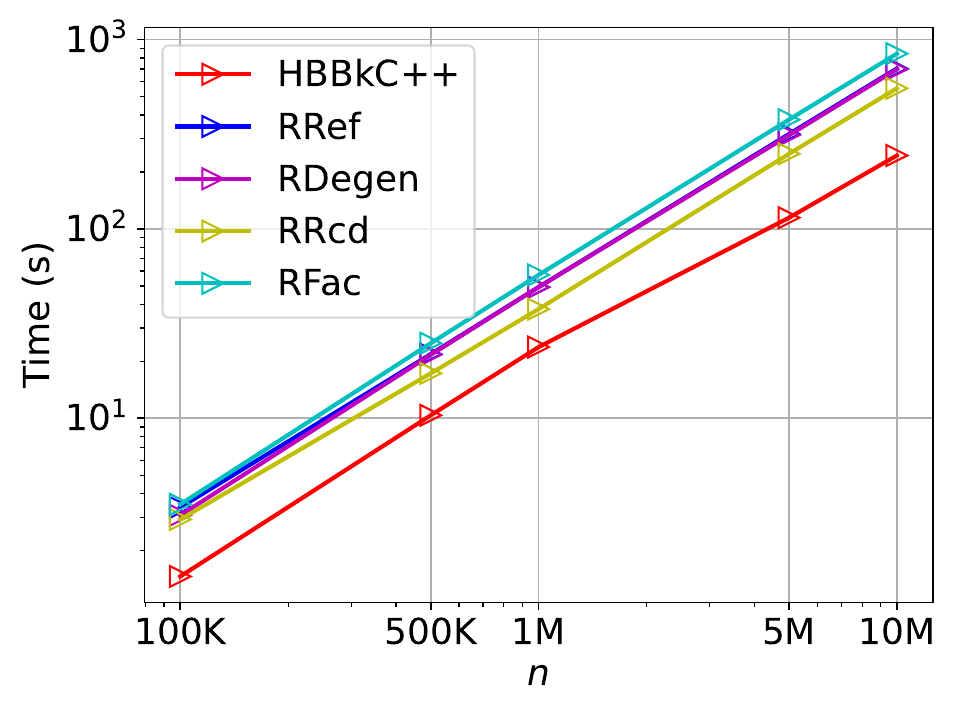}}
\subfigure[Results on varying $\rho$ (ER model).]{
\label{fig:density-er}
\includegraphics[width=0.23\textwidth]{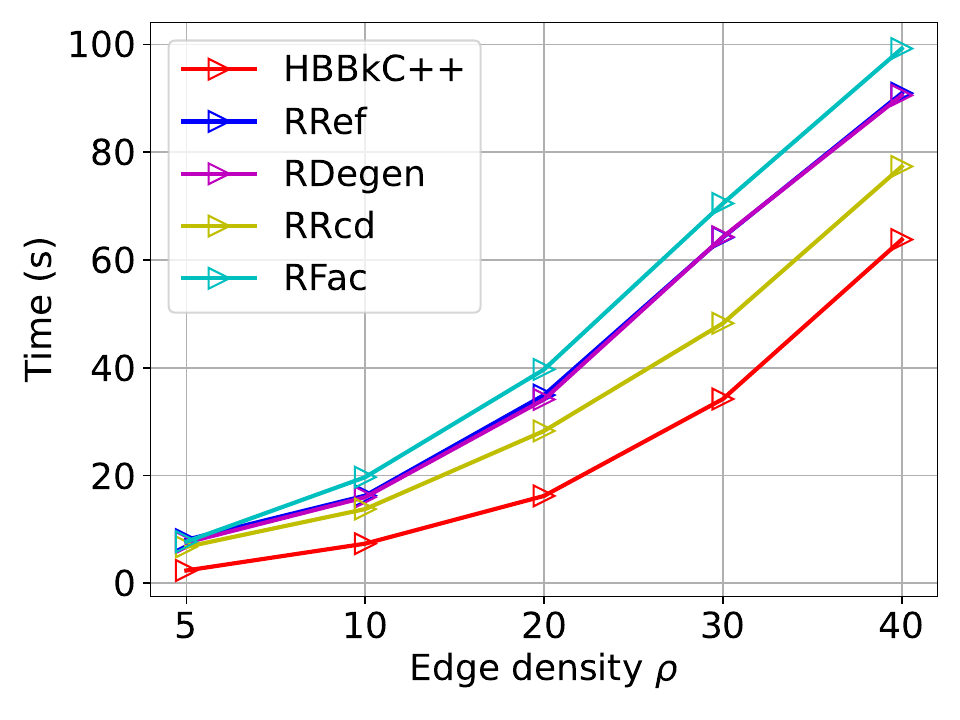}}
\subfigure[Results on varying $\rho$ (BA model).]{
\label{fig:density-ba}
\includegraphics[width=0.23\textwidth]{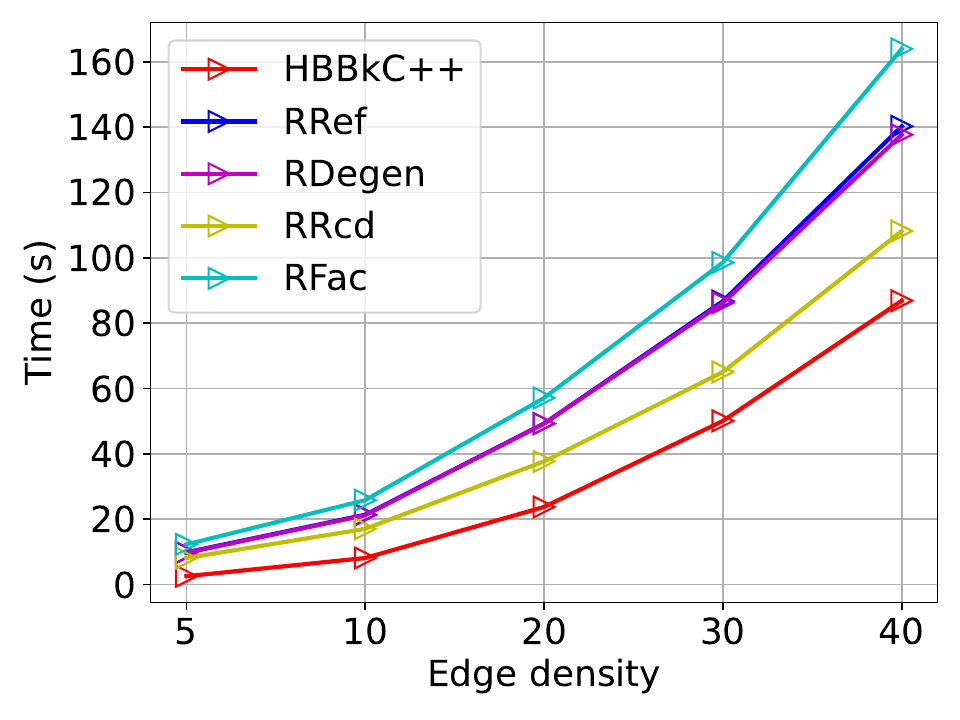}}
\caption{ Results on synthetic datasets. }
\end{figure}

\smallskip\noindent\textbf{Results.}
First, we vary the number of vertices in the range $n\in\{100\text{K}, 500\text{K}, 1\text{M}, 5\text{M}, 10\text{M}\}$ and set $\rho=20$ {\cheng by} default. The results for the graphs generated using the ER model and the BA model are presented in Figure~\ref{fig:scale-er} and Figure~\ref{fig:scale-ba}, respectively. Under all settings, \texttt{HBBMC++} consistently outperforms all baselines, achieving remarkable speedups compared to the second-best algorithm \texttt{RRcd}, with average speedups of 1.84x and 1.94x (up to 2.84x and 2.26x) on the graphs generated by the ER and BA models, respectively.
Additionally, we report the degeneracy number $\delta$ and the truss-related number $\tau$ under different settings. In the graphs generated by the ER model, these values are $\delta=\{8, 14, 28, 43, 58\}$ and $\tau=\{0, 1, 3, 5, 7\}$ for different values of $n$. In contrast, for the graphs generated by the BA model, the values remain approximately $\delta\approx 42$ and $\tau\approx 25$ across all $n$. 
This interesting phenomenon can be explained by the distinct mechanisms of edge formation in the two models. In the ER model, edges are uniformly distributed, which facilitates the formation of larger  $k$-cores as the number of vertices $n$ increases, leading to an increase in degeneracy. Conversely, in the BA model, the preferential attachment mechanism concentrates edges around a few high-degree hubs, stabilizing the graph’s core structure and maintaining a nearly constant degeneracy, irrespective of the growth in $n$.
Finally, we note that all above results are well aligned with our theoretical analysis, i.e., the worst-case time complexity of \texttt{HBBMC++} is better than that of existing algorithms when $\delta\ge \tau +\frac{3}{\ln 3}\ln \rho$.

Second, we vary the edge density in the range $\rho\in\{5,10,20,30,40\}$ and set $n=1\text{M}$ {\cheng by} default. The results for the graphs generated using the ER and BA models are presented in Figure~\ref{fig:density-er} and Figure~\ref{fig:density-ba}, respectively. 
On average, \texttt{HBBMC++} achieves speedups of 1.65x and 1.88x (with maximum speedups of 1.89x and 3.17x) on the graphs generated by the ER and BA models, respectively, compared to the second-best algorithm \texttt{RRcd}. We also observe that the speedup decreases as the graph $\rho$ increases. The possible reason is that when $\rho$ increases, the graph becomes denser. Therefore, the differences between $\delta$ and $\tau$ becomes {\cheng narrower}.

Notably, we observe that when comparing the results under identical parameter settings but using the ER model, the computation time is longer for graphs generated by the BA model. The possible reason is that the graphs following a power-law distribution typically contain larger sizes of maximal cliques. For example, under the default setting (i.e., $n=1\text{M}$ and $\rho=20$), the largest clique in the ER graph is a 5-clique, while for the BA graph, the largest one is a 28-clique.

\balance
\bibliographystyle{IEEEtran}
\bibliography{reference}

\end{document}